%% file: venn-mlsys25.tex
\documentclass{article}
\usepackage{microtype}
\usepackage{graphicx}
\usepackage{subcaption}
\usepackage{booktabs} % for professional tables
\usepackage{amssymb}

\usepackage{hyperref}

\include{amber-common}

\usepackage[accepted]{lib/mlsys2025} 

\hypersetup{
  citecolor=.
}

\def\name{{Venn}\xspace}
\def\fl{{CL}\xspace}

\mlsystitlerunning{{\name}: Resource Management for Collaborative Learning Jobs}

\usepackage{graphicx}
\usepackage{hyperref}
\usepackage{eso-pic}
\usepackage{xparse}
\usepackage{etoolbox}  % Add this line

\newlength{\badgewidth}
\newlength{\badgegap}
\newcommand{\badgeList}{}

\NewDocumentCommand{\addTopRightBadge}{O{} m}{%
\gappto{\badgeList}{\href{#1}{\includegraphics[width=\badgewidth]{#2}}\hspace{\badgegap}}%
}

\newcommand{\placeTopRightBadges}{%
\AddToShipoutPictureBG*{%
\put(\LenToUnit{\paperwidth - 1.5cm - \badgewidth},\LenToUnit{\paperheight - 2cm}){%
\makebox[0pt][r]{\badgeList}%
}%
}%
}
\addTopRightBadge{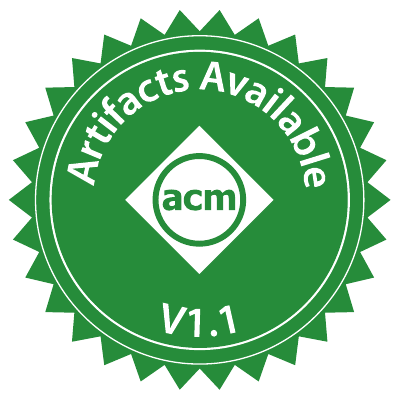}
\addTopRightBadge{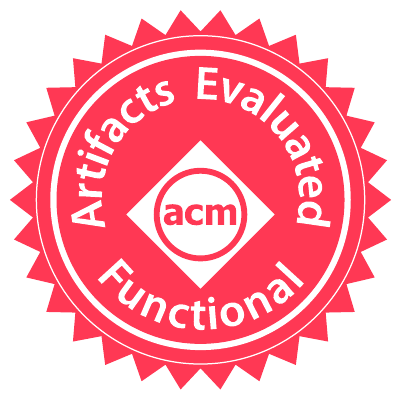}
\placeTopRightBadges
\begin{document}
% \acmBadgeR{figures/artifacts-available-v1.1.pdf}

\twocolumn[
\mlsystitle{{\name}: Resource Management for Collaborative Learning Jobs}

% unless you've provided the [accepted] option to the mlsys2021


%\mlsyssetsymbol{equal}{*}
%
\begin{mlsysauthorlist}
\mlsysauthor{Jiachen Liu}{umich}
\mlsysauthor{Fan Lai}{uiuc}
\mlsysauthor{Ding Ding}{umich}
\mlsysauthor{Yiwen Zhang}{umich}
\mlsysauthor{Mosharaf Chowdhury}{umich}
%\mlsysauthor{Paper \#?, \pageref{EndOfPaper} Pages}
\end{mlsysauthorlist}
 
\mlsysaffiliation{umich}{  Computer Science and Engineering,
University of Michigan, Michigan, USA}
\mlsysaffiliation{uiuc}{ Computer Science and Engineering,
University of Illinois at Urbana-Champaign, Illinois, USA}

\mlsyscorrespondingauthor{Jiachen Liu}{amberljc@umich.edu}

%\mlsyskeywords{Machine Learning, MLSys}
%\printAffiliationsAndNotice{\mlsysEqualContribution} % otherwise use the standard text.
\begin{abstract}
\input{pages/abstract.tex}
\end{abstract}
]



\printAffiliationsAndNotice{}  % leave blank if no need to mention equal contribution
%\printAffiliationsAndNotice{\mlsysEqualContribution} % otherwise use the standard text.

\input{pages/intro.tex} 
\input{pages/motivation.tex}

\input{pages/overview.tex}
\input{pages/algorithm.tex}

\input{pages/eval.tex}

\input{pages/related.tex}

\input{pages/conclusion.tex}

\input{pages/ack.tex}
\label{EndOfPaper}

\bibliography{venn}
\bibliographystyle{lib/mlsys2025}

\clearpage
\appendix

\input{pages/appendix.tex}

\clearpage

%\appendix

\end{document}

%% file: amber-common.tex
\usepackage[table,xcdraw]{xcolor}% If accepted, instead use the following line for the camera-ready submission:

\usepackage{amsmath,amssymb,amsthm}

\usepackage{xspace}
\def\ie{{i.e.\xspace}}
\def\eg{{e.g.\xspace}}

\def\etc{etc.\xspace}

\usepackage{mathtools}

\usepackage{makecell}

\usepackage{subcaption}

\usepackage{tikz}

\graphicspath{{figures/}}

\usepackage[inline]{enumitem}
\newlength{\docparskip}
\setlist[enumerate]{
    nosep,
    topsep=2pt,
    partopsep=4pt,
    itemsep=2pt,
    leftmargin=1.5em,
}
\setlist[itemize]{
    nosep,
    topsep=2pt,
    partopsep=4pt,
    itemsep=2pt,
    leftmargin=1.5em,
}
\setlist[description]{
    nosep,
    topsep=2pt,
    partopsep=4pt,
    itemsep=2pt,
}
\newlist{inlineenum}{enumerate*}{1}
\setlist[inlineenum]{label=\itshape\alph*\upshape)}

%% file: pages/abstract.tex
In recent years, collaborative learning (\fl) has emerged as a promising approach for machine learning (ML) and data science across distributed edge devices. 
As the deployment of CL jobs increases, they inevitably contend for limited resources.
However, efficient resource scheduling in this context is challenging because of the \emph{ephemeral nature and resource heterogeneity of devices}, coupled with the \emph{overlapping resource requirements of diverse CL jobs}.
Existing resource managers often assign devices to CL jobs randomly for simplicity and scalability, but this approach compromises job efficiency.

In this paper, we present \name, a \fl resource manager that efficiently schedules ephemeral, heterogeneous devices among multiple \fl jobs to reduce the average job completion time (JCT).
\name formulates the \emph{Intersection Resource Scheduling (IRS)} problem to identify complex resource contention among multiple \fl jobs.
It then proposes a contention-aware scheduling heuristic to minimize the average scheduling delay. 
Furthermore, it proposes a resource-aware device-to-job matching heuristic to optimize response collection time by mitigating stragglers.
Our evaluation shows that, compared to the state-of-the-art CL resource managers, \name improves the average JCT by up to $1.88\times$.
The code is available at \url{https://github.com/SymbioticLab/Venn}.

%% file: pages/intro.tex
\section{Introduction}%
\label{sec:intro} 
 
Collaborative learning (\fl) enables distributed edge devices to perform collaborative machine learning (ML) without moving raw data into the cloud~\cite{google, fa}. 
\fl has been adopted by many large corporations including Apple, Meta, Google, and LinkedIn to protect user data privacy while improving user experience. 
For example, Google adopts \fl for a wide range of applications such as speech recognition~\cite{google-speech}, healthcare study~\cite{google-health}, next-word prediction~\cite{google-keyboard, gboard-fl}, emoji prediction~\cite{google-emoji}, and query suggestion on keyboard~\cite{google-query-sugg}. 
Each \fl training job in practice often requires 1000$\sim$10000 device participants in each training round and takes 4$\sim$8 days to finish~\cite{google-query-sugg}.
As the number of \fl applications continues to grow, efficient edge resource management has become the key to fast and resource-efficient \fl. 

Unlike cloud resources, \fl resources are not accessible all the time; they only become available for training when they are connected to WiFi and charging~\cite{google}.
Moreover, they are highly heterogeneous in terms of hardware capacity, software versions, and training data availability (\S\ref{sec:fl}).
Even when running on the same client population, different \fl jobs often compete for different subsets of devices, based on their specific model characteristics and training objectives. 
The resource heterogeneity and job's resource requirement diversity lead to complex multi-resource contentions among ongoing \fl jobs, where eligible resources for each job are not only limited but also may \emph{overlap, contain, or be within} those of one or more other jobs.

\begin{figure}[t!]
   \centering 
      \includegraphics[trim=0 140 30 0,clip,scale=1.1]{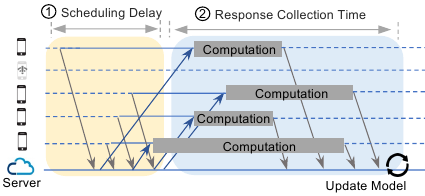} 
     \caption{Composition of the completion time of one round of a \fl job.} 
   \label{fig:jct}    
\end{figure}

However, existing \fl efforts often overlook resource contention among coexisting \fl jobs, assuming that sufficient resources are always available~\cite{google}.
For example, many recent works~\cite{fedscale, relay, oort, divfl} primarily focus on optimizing the response collection time of individual \fl jobs.
In practice, however, multiple \fl jobs may run concurrently and compete for a subset of devices, leading to resource contention~\cite{google}.
As a result, existing works fail to account for scheduling delay (\S\ref{sec:flsys})---the time needed to acquire all necessary resources for a \fl job---as a key factor influencing job completion time (JCT) (Figure~\ref{fig:jct}).

With the proliferation of \fl in production, large corporations such as Apple~\cite{apple}, Meta~\cite{meta}, and Google~\cite{google} have developed their own \fl infrastructures to manage multiple jobs at scale.
However, despite their low-level differences, these \fl resource managers can largely be characterized as \emph{random device-to-job matching} in various forms.
We observe that such random matching can result in higher scheduling delays and longer response collection times as an increasing number of \fl jobs compete for a shared pool of devices.

In this paper, we present \name, a \fl resource manager designed to minimize the average JCT of multiple \fl jobs competing for a large pool of edge resources.
First, \name models the complex resource contention among \fl jobs as an instance of the \emph{Intersection Resource Scheduling (IRS)} problem (§\ref{sec:irs}), where a job's resource demands may overlap with, contain, or fall within those of other jobs.
\name then introduces a contention-aware scheduling heuristic that prioritizes jobs requiring scarce resources and lower overall demand to minimize the average scheduling delay.
Second, \name employs a resource-aware device-to-job matching heuristic to mitigate the effect of stragglers on response collection time (\S\ref{sec:match}).
Together, these techniques jointly optimize the average JCT of multiple \fl jobs operating under dynamically changing and uncertain resource conditions.

We have implemented and evaluated \name across various \fl workloads derived from real-world scenarios (\S\ref{sec:eval}). 
Compared to state-of-the-art \fl resource allocation solutions \cite{google, meta, apple}, \name improves the average JCT by up to $1.88\times$.

Overall, we make the following contributions in this paper:
\begin{enumerate}
   \item We introduce \name, a \fl resource manager designed to enable efficient sharing of heterogeneous devices across a large number of \fl jobs.
   
   \item To minimize the average JCT of \fl jobs, we propose a job scheduling and client matching joint solution to optimize both the scheduling delay and response collection time.
   Our approach achieves a logarithmic-linear time complexity and is supported by theoretical analysis.

   \item We have implemented and evaluated \name, along with its scheduling and matching algorithms, demonstrating improvements in the average JCT compared to the state-of-the-art across various real-world \fl workloads.
\end{enumerate}

%% file: pages/motivation.tex
\section{Background and Motivation}
\label{sec:background} 

\subsection{Collaborative Learning}
\label{sec:fl}
Collaborative learning (\fl) has been widely adopted in industry to train ML models on diverse datasets stored on edge devices without transferring raw data to the cloud.
Example applications include healthcare study, speech recognition, next-word prediction, \etc
However, \fl introduces unique resource management challenges due to the distinct characteristics of edge devices and the jobs that rely on them.

\begin{figure}[t!]
  \centering
     \begin{subfigure}{1.5in}  
   \centering 
      \includegraphics[ scale=0.295]{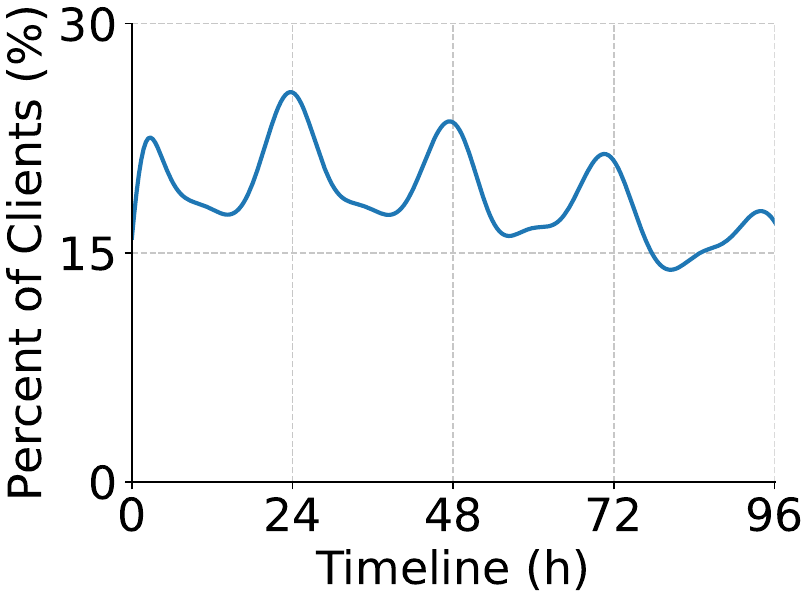} 
     \caption{Diurnal device availability. } 
   \label{fig:avail}
     \end{subfigure} 
  \hfill
  \begin{subfigure}{1.7in} 
  \centering
      \includegraphics[ scale=0.3]{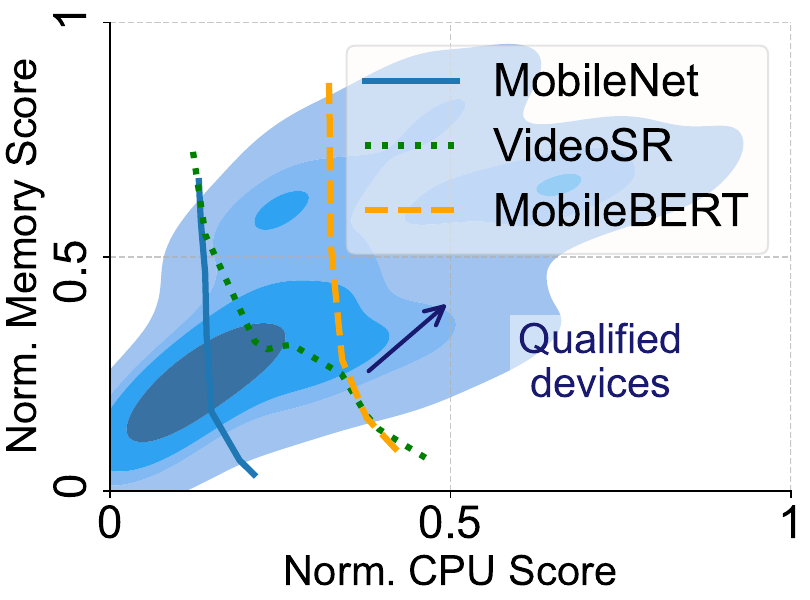} 
      \caption{Device hardware heterogeneity.}
   \label{fig:device-dist}
 \end{subfigure}
    \caption{\fl resources exhibit both high variance in availability and capacity.} 
   \label{fig:client}
\end{figure}

\paragraph{Single \fl Job's Resource Management:}
\fl resources often refer to edge devices with limited capacities, including smartphones, laptops, and Internet of Things (IoT) devices.
A major challenge of resource management in \fl stems from the unique characteristics of \fl resources in terms of availability and heterogeneity.

\noindent\emph{Dynamic availability.}
Unlike cloud resources (\eg, GPUs) that are accessible on demand, edge devices in \fl are only available for training under specific conditions, such as when they are charging and connected to WiFi.
We analyze a real-world client availability trace~\cite{fedscale}, which encompass 180 million trace items of devices behavior over a week. 
Figure~\ref{fig:avail} shows that the number of available devices (charging and connected to WiFi) changes over time with diurnal pattern.
This fluctuation creates inherent scheduling delays, as \fl jobs must wait for enough devices to become available before training can begin.
% a sufficient number of device resources.

\noindent\emph{Device heterogeneity.}
Edge devices in \fl vary widely in hardware capabilities, including memory and CPU capacity, as well as in data availability and software version, leading to inconsistent response times.
Figure~\ref{fig:device-dist} showcases this heterogeneity, focusing on variations in memory and CPU capacities among edge devices, based on data from AI Benchmark~\cite{ai-bench}.  
It also annotates the minimum hardware requirements needed to execute three popular on-device ML models within a reasonable time.  
Additionally, \fl jobs often target specific device subsets based on factors like data availability, software version, or hardware capacity.
These subsets can \emph{overlap, nest, or be mutually exclusive}, leading to complex resource contention patterns where jobs compete for the same devices.

Together, dynamic availability and device heterogeneity slow down \fl training, prompting solutions like client selection~\cite{oort,relay,divfl}, quantization~\cite{quant-fl}, and model aggregation~\cite{gluefl} to optimize single-job performance.

\begin{figure}[t!] 
   \begin{subfigure}{0.44\textwidth}
      \centering
   \includegraphics[clip, trim=0 55 0 40,  scale=0.21]{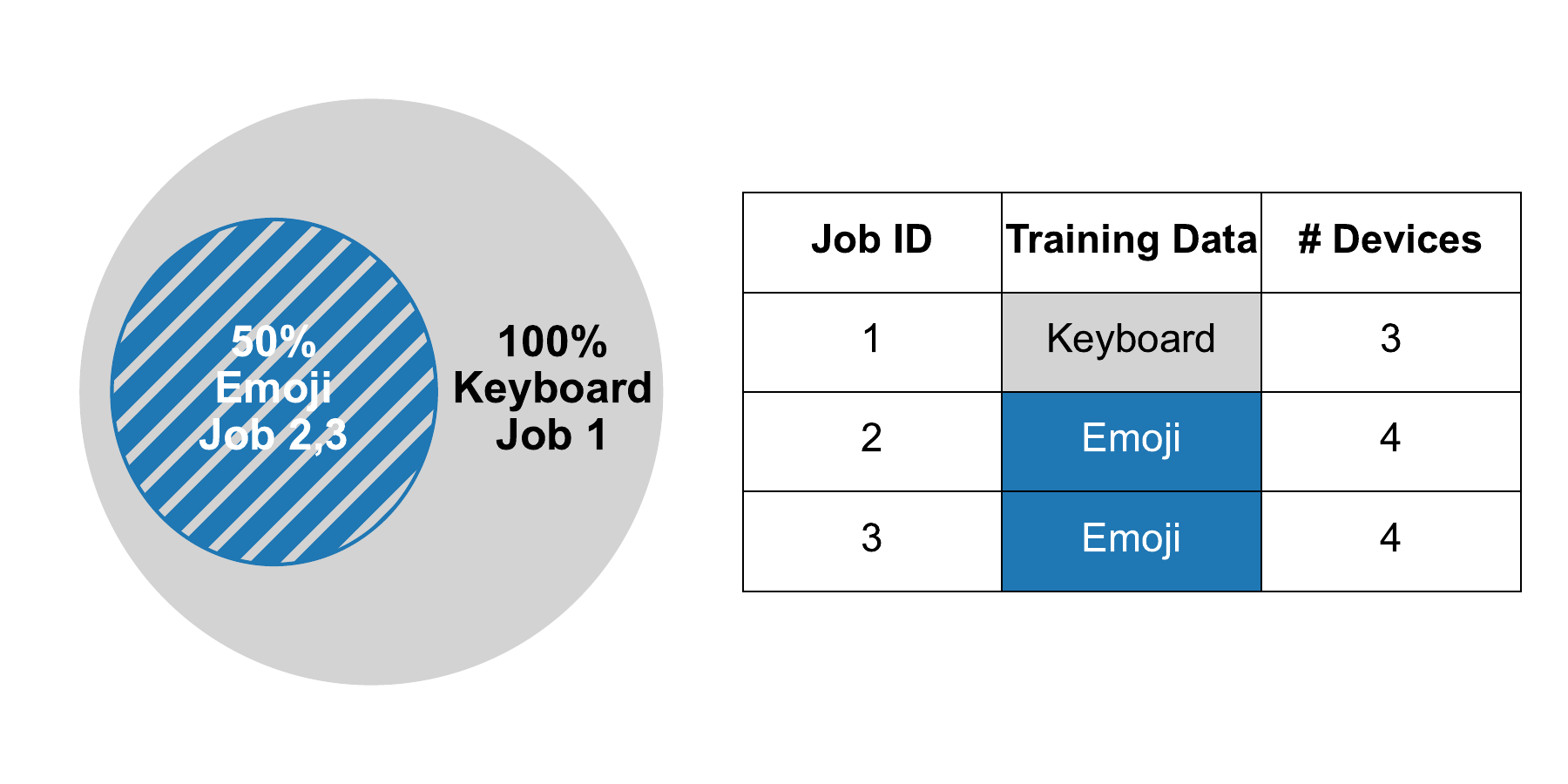}
   \caption{Setup.}
      \label{fig:toy-setup}
   \end{subfigure}   
   
   \begin{subfigure}{0.5\textwidth}
      \includegraphics[clip, trim=10  30 0 20,  scale=0.27]{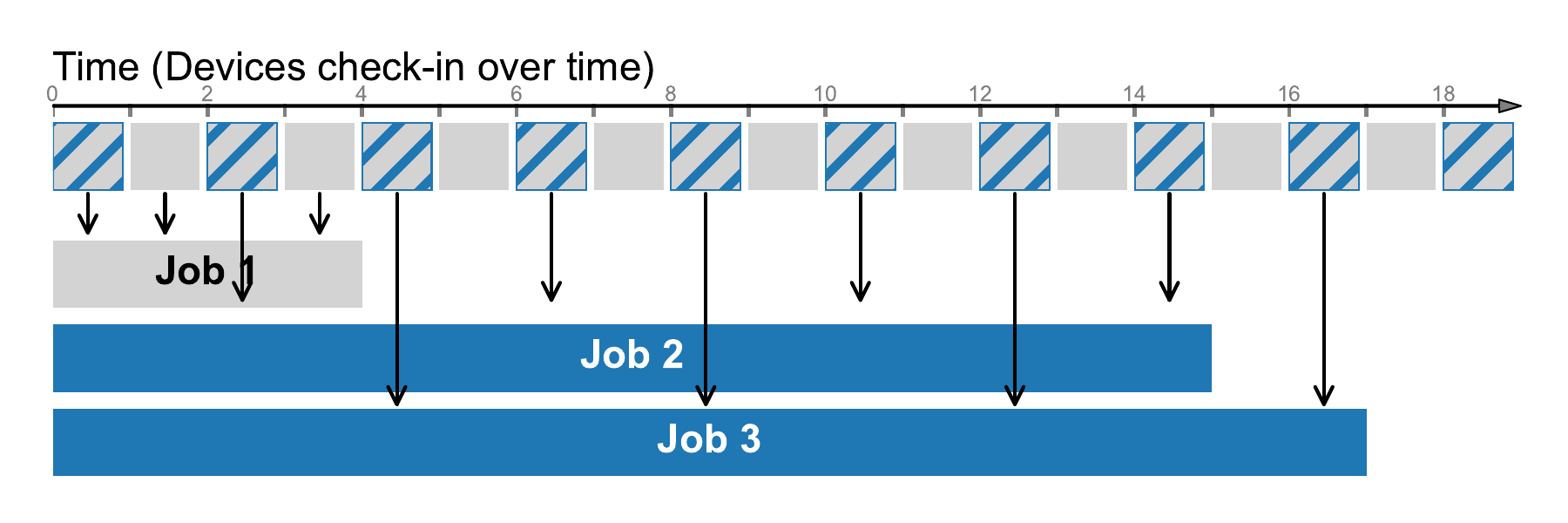}
      \caption{Random Matching ($\bar{JCT}$=12).}
      \label{fig:random}
   \end{subfigure} 
   \begin{subfigure}{0.5\textwidth}
      \includegraphics[clip, trim=10  30 0 20,  scale=0.27]{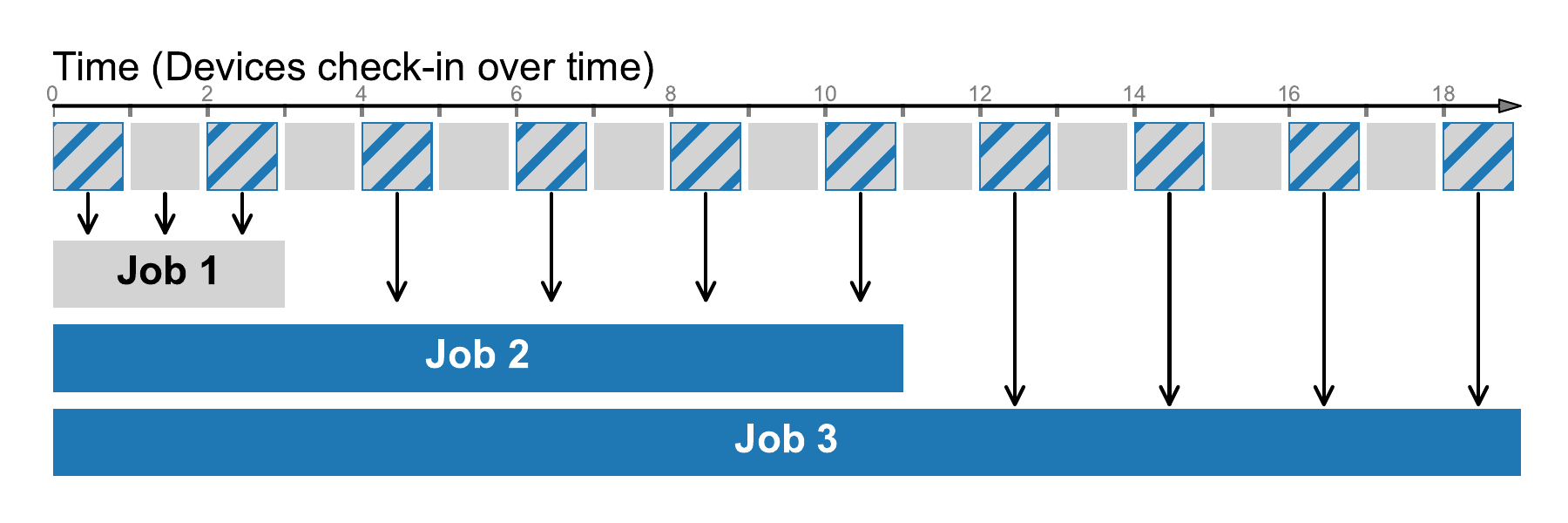}
      \caption{Shortest remaining service first (SRSF) ($\bar{JCT}$=11).}
      \label{fig:srsf}
   \end{subfigure} 
   \begin{subfigure}{0.5\textwidth}
      \includegraphics[clip, trim=10 30 0 20,  scale=0.27]{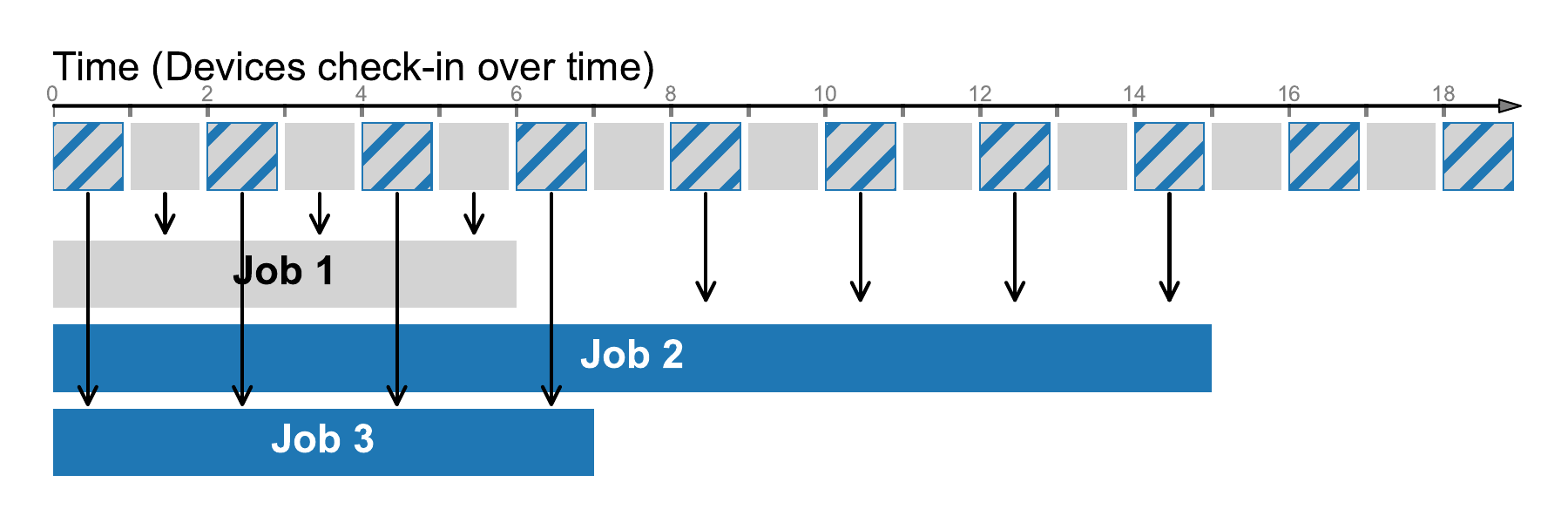}
   \caption{Optimal ($\bar{JCT}$=9.3).}
   \end{subfigure}
       \caption{Toy example showing three resource schedules across multiple CL jobs.
       Job demands and resource eligibility are shown in the top row.
       Devices check in at a constant rate. 
       Eligible devices only for Emoji jobs are marked with blue; all devices are eligible for the Keyboard job.
       The label of each client indicates its job assignment. 
      Random Matching and SRSF inefficiently allocate scarce Emoji-eligible devices to Job 1, which already has sufficient Keyboard-eligible resources.
       In contrast, the optimal schedule allocates these scarce resources to Job 2 followed by Job 3, minimizing the average JCT.
       }
      \label{fig:toy}
   \end{figure} 

\subsection{Multiple \fl Jobs' Resource Management}
\label{sec:flsys}

To orchestrate multiple \fl jobs, several \fl resource managers have been proposed with three primary designs.  
\begin{enumerate}%[label=(\roman*)]
    \item Apple's \fl resource management~\cite{apple} is driven by clients, where each client independently samples from a list of \fl jobs they are able to execute.
    
    \item Meta's \fl resource manager~\cite{meta} is centralized, where it randomly matches each client with one eligible \fl job.
    
    \item Google's \fl resource manager~\cite{google} is driven by jobs, where each job independently samples from available clients. 
\end{enumerate}

Despite the seeming variety in their designs, existing resource managers all rely on random device-to-job matching.
This approach works when devices are abundant but fails to address resource contention when demand exceeds supply, resulting in longer scheduling delays and inefficient resource use.

\subsection{Limitation and Opportunities}

\label{sec:motivation}

\paragraph{Limitations of the state-of-the-art.}
To highlight the shortcomings of existing methods, consider the toy example in Figure~\ref{fig:toy}, which compares three scheduling strategies for three \fl jobs—“Emoji Prediction Job 1,” “Emoji Job 2,” and a “Keyboard Prediction Job”—competing for devices with varying eligibility.
% Consider the toy example in Figure~\ref{fig:toy} that compares three scheduling solutions: random matching as specified above, Shortest-Remaining-Service-First (SRSF) proposed for cloud ML scheduling \cite{tiresias}, and the optimal solution.
We assume all jobs arrive simultaneously and clients with different eligibilities become online constantly over time.
Both random matching and SRSF waste scarce resources (\ie, Emoji clients) on the Keyboard job, which has ample device options, while the optimal schedule prioritizes efficiency, reducing JCT to 9.3 time units.
Note that the contention patterns in real-world scenarios are often more complex and larger in scale than the one presented in this example.

\begin{figure}[t!]
   \centering 
     \begin{minipage}{1.5in}
      \includegraphics[trim=0 0 0 0,clip,scale=0.3]{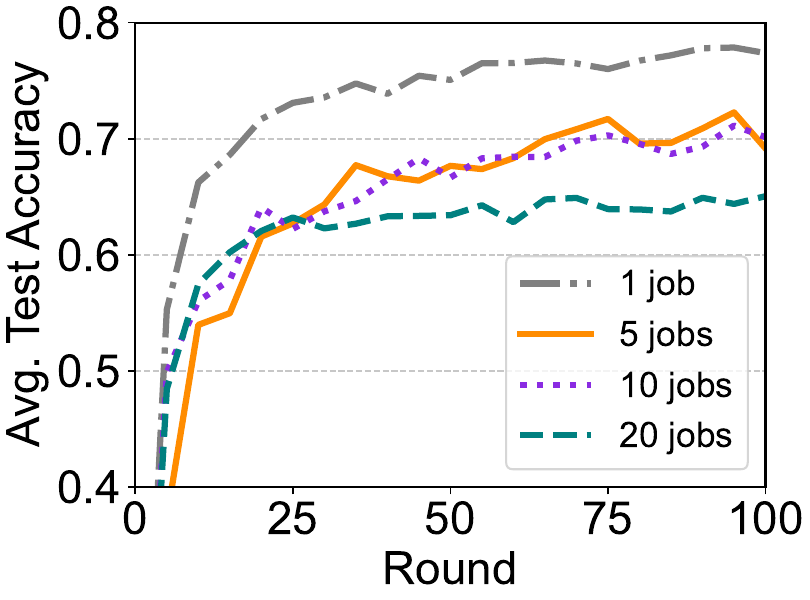} 
     \caption{Impact of resource contention. } 
   \label{fig:contention} 
        \end{minipage} 
      \hfill
     \begin{minipage}{1.5in}
      \includegraphics[scale=0.3]{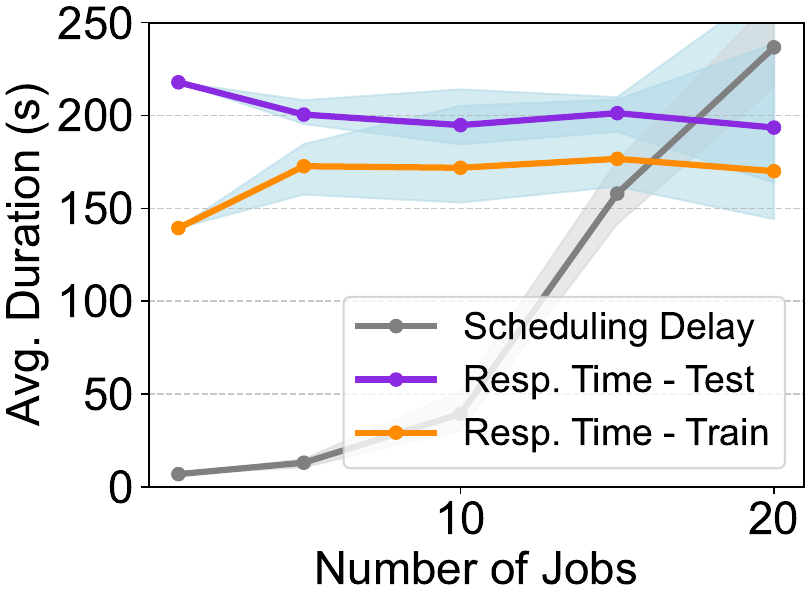} 
     \caption{JCT breakdown in a single round.} 
   \label{fig:jct-breakdown} 
        \end{minipage} 
\end{figure}

\paragraph{Impact of resource contention.} 
Single-job FL optimizations, such as client selection~\cite{oort,relay,divfl}, often assume that \fl jobs have access to sufficient number of online devices at any time.
However, in practice, there could be multiple \fl jobs running at the same time and competing for the same set of devices, leading to resource contention~\cite{google}.
This can significantly impact the performance of \fl jobs, such as accuracy and end-to-end training time.

%\amber{can be improved}
We analyze the impact of resource contention on \fl jobs' performance in Figure~\ref{fig:contention}.
In this experiment, the resource pool is evenly partitioned and managed by each job, who aims to train a ResNet-18 model with 100 clients per round on the FEMNIST dataset~\cite{FEMNIST}.
We vary the number of concurrently running jobs to observe its impact.
As more jobs share the same device pool, the available device choices for each job become increasingly constrained, leading to a noticeable degradation in the round-to-accuracy performance.
Hence, evenly partitioning resource pool for each \fl job is no better than a shared device pool in terms of participant diversity.

\paragraph{Breakdown of request completion time.}  
While most existing \fl efforts, such as quantization~\cite{quant-fl} and client selection~\cite{gluefl}, predominantly focus on model convergence rate or \emph{response collection time}—--\ie, the time needed to collect a sufficient number of responses—--they often overlook a critical component: \emph{scheduling delay}, as depicted in Figure~\ref{fig:jct}.

Figure~\ref{fig:jct-breakdown} breaks down JCT for a single training round under varying contention levels, using the same setup as Figure~\ref{fig:contention}.
Utilizing the same experimental setup as in Figure~\ref{fig:contention}, we quantify both the average scheduling delay and response collection time with random device-to-job matching during one round of training and testing.  
The shaded regions cover the duration for each individual job. 
As our results in Figure~\ref{fig:jct-breakdown} indicate, scheduling delay can significantly impact overall JCT, especially when resource supply falls short of demand.

%% file: pages/overview.tex
\section{\name Overview}
\label{sec:overview}  

\name serves as a standalone \fl resource manager that operates at a layer above all \fl jobs, and it is responsible for allocating each checked-in resource to individual jobs.
Figure~\ref{fig:amg-overview} illustrates \name's workflow and its role within the lifecycle of an \fl job. 
 
 \begin{figure}[t!]
    \centering
    \includegraphics[trim=0 95 90 0,clip,scale=1.2]{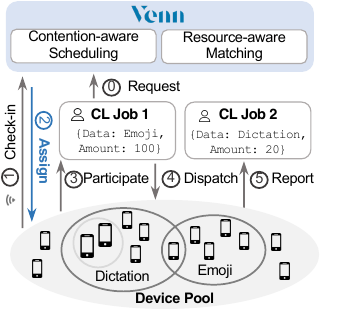} 
      \caption{\name system overview.}
      \label{fig:amg-overview}
\end{figure}

In each execution round, a job submits resource requests to \name, specifying its device requirements (e.g., minimum CPU capacity) and resource demands (the number of devices needed per round) (\textcircled{0}). 
Meanwhile, edge devices continuously check in with \name as they become available over time (\textcircled{1}).
Based on the real-time resource demand and supply, \name generates a resource allocation plan to assign one \fl job to each checked-in device until the job's needs are met (\textcircled{2}).
Upon receiving the task assignment from \name, each device adheres to the allocation plan and participates in the corresponding job (\textcircled{3}).
The device then retrieves the computation plan from the job and performs on-device training or inference~\citep{walle,coreml}  (\textcircled{3} and \textcircled{4}). 
Finally, devices either submit their training results to the job or disconnect if their availability changes (e.g., battery runs low) (\textcircled{5}).

Steps \textcircled{3} to \textcircled{5} follow standard \fl protocols between jobs and devices, while \name focuses on optimizing the job-to-device assignment in Step \textcircled{2}.
This allocation step is critical for managing resource contention and reducing scheduling delays.
A detailed breakdown of responsibilities is provided in Appendix~\ref{sec:app-resp}.

%% file: pages/algorithm.tex
\section{Resource Scheduling in \name}
\label{sec:algo}

In this section, we outline the resource scheduling algorithm in \name, starting with defining the scheduling problem in \fl (\S\ref{sec:prob}).
Next, we present our scheduling algorithm in two parts: determining the job scheduling order to minimize scheduling delay (\S\ref{sec:irs}) and matching devices to jobs to optimize response collection time (\S\ref{sec:match}).
Finally, we describe enhancements for real-world deployment (\S\ref{sec:enhancement}).

\subsection{Problem Statement}
\label{sec:prob}
Given a collection of \fl jobs---along with their device requirements and resource demands---and a set of heterogeneous devices that arrive and depart over time, \name should efficiently assign devices to \fl jobs in order to reduce the average job completion time (JCT). 
The scheduling problem can be mathematically modeled as a multi-commodity flow (MCF) problem with integer constraints, where each \fl job is modeled as a distinct commodity and each device serves as an intermediate vertex between the source and sink of its corresponding eligible \fl job. 
Then the goal of this integer MCF problem is to minimize the average JCT of jobs, which is known to be NP-hard~\citep{multi-commodity}.
Even for its linear approximation, the time complexity is exacerbated by the planetary scale of devices involved and diverse device requirements from jobs, making existing solutions computationally infeasible.
% We now formally define such resource scheduling problem.

\paragraph{Problem definition.}
Let $\mathbb{J} = {J_1, J_2, \ldots, J_m}$ be the set of jobs, with resource demands $\mathbb{D} = {D_1, D_2, \ldots, D_m}$, where $D_i$ is the number of devices required by job $J_i$.
% Assume we have $m$ jobs $\mathbb{J}= \{J_1,J_2, \allowbreak ..., J_{m} \}$ with their resource demands $\mathbb{D}= \{D_1,D_2, \allowbreak ..., D_{m} \}$.
Let $\mathbb{S} = S_1 \cup S_2 \cup \cdots \cup S_n$ be the set of available devices, where $S_k$ is the subset of devices eligible for jobs with specific requirements (e.g., hardware capacity), and $f(J_i) = S_k$ maps job $J_i$ to its eligible device subset.
The goal is to match each checked-in device $s \in S_k$ with one job $J_i$, where $f(J_i)=S_k$, $\forall s \in \mathbb{S}$, in order to minimize average JCT, which consists of scheduling delay and response collection time.

\paragraph{Tradeoff between scheduling delay and response collection.}
Minimizing scheduling delay often conflicts with minimizing response collection time.
Assigning devices quickly reduces delay but may involve slower devices, increasing response time, since the response collection time is usually determined by the final responding device among the target number of participants.
Waiting for faster devices, however, reduces response time at the cost of higher delay.

At its core, \name aims to find a sweet spot in the trade-off in order to optimize average JCT. Specifically, we decouple the \fl resource allocation problem as following two questions:
\begin{enumerate} 
  \item \emph{How should jobs be ordered to minimize scheduling delay? (\S\ref{sec:irs})}
  \item \emph{How should devices be matched to jobs to minimize response collection time while optimizing JCT? (\S\ref{sec:match})}
\end{enumerate}
 
\definecolor{egyptianblue}{rgb}{0.06, 0.2, 0.65}
\newcommand{\comm}[1]{\textcolor{egyptianblue}{#1}}

\begin{algorithm}[t!]
    \caption{Intersection Resource Scheduling}
    \label{algo:amg-sched}
    \begin{algorithmic}[1]
        \STATE {\bfseries Input:} Job Groups $\mathbb{G}$, Devices $\mathbb{S}$
        
        \textbf{Function}\textit{\name-Sched}{Job Groups $\mathbb{G}$, Devices $\mathbb{S}$}
        
       $\triangleright$ {\comm{Sort within job group (\S\ref{sec:within})}}
        \FOR{$G_j$ in $G$}
            \STATE Sort $J_i$ by $D_i$ in ascending order, $\forall J_i \in G_j$ \label{code:sort}
        \ENDFOR
        
  $\triangleright${\comm{Generate initial allocation $S_j'$ for each group $G_j$}}
        \STATE $S = \cup_{j=1}^n S_j$ \label{code:init-start}
        \STATE Sort $G_j$ by $|S_j|$ in ascending order, $\forall G_j \in \mathbb{G}$ \label{code:across}
        \FOR{$G_j$ in $G$}
            \STATE $S_j' = S \cap S_j, S = S \setminus S_j'$
        \ENDFOR \label{code:init-end}
        
 $\triangleright${\comm{Allocate resource $S_j'$ for each group $G_j$}}
        \STATE Sort $G_j$ by $|S_j|$ in descending order, $\forall G_j \in \mathbb{G}$ \label{code:check}
        \FOR{$G_j$ in $G$} \label{code:loop-start}
            \IF{$|S_j'| > 0$} \label{code:not-none}
                \FOR{$G_k \in \mathbb{G}: |S_k| < |S_j|, S_k \cap S_j \neq \emptyset$}
                    \STATE $m_j', m_k'$ = \texttt{get-queue-len}() \label{code:queue-len}
                    \IF{$\frac{m_j'}{|S_j'|} > \frac{m_k'}{|S_k|}$} \label{code:condition}
                        \STATE $S_j' = S_j' \cup (S_j \cap S_k)$ \label{code:update-start}
                        \STATE $S_k' = S_k' - S_j'$ \label{code:update-end}
                    \ELSE
                        \STATE break \label{code:break}
                    \ENDIF
                \ENDFOR
            \ENDIF \label{code:across-end}
        \ENDFOR \label{code:loop-end}
        
        \STATE \textbf{Return} $\{G_j[0], S_j'\}, \forall j \in [1,n]$
         
    \end{algorithmic}
\end{algorithm}

\subsection{Intersection Resource Scheduling (IRS)}
\label{sec:irs} 
We now tackle the first question, minimizing the scheduling delay. Directly matching jobs to devices can be computationally expensive, especially when dealing with the immense scale of devices and jobs. 
The challenges posed by this problem are not solely due to its scale; they are also compounded by the diverse resource requirements of \fl jobs. 
Their various requirements introduce intricate resource contention patterns, further complicating the scheduling. 

To this end, \name introduces the Intersection Resource Scheduling (IRS) problem to account for this resource contention.  
Basically, each \fl job $J_ i \in \mathbb{J}$ may compete for a subset of devices $S_k \in \mathbb{S}$, denoted as $f(J_i)=S_k$, where these resource subsets can exhibit relationships that are \emph{inclusive, overlapping, or nested}.
We created an integer linear programming (ILP) formulation (Appendix~\ref{sec:ilp}) to optimally allocate resources to minimize scheduling delay and propose a heuristic to tackle the problem.

To tackle the scale of devices and jobs, \name aims to determine a job scheduling order, where each checked-in device is assigned to the first eligible job in the order, rather than scattering resources across multiple jobs.
Such a fixed job order can minimize the scheduling delay while reducing computational complexity.

With the objective of determining a job scheduling order, \name first groups jobs $\mathbb{J}$ into \emph{Resource-Homogeneous Job Groups} $\mathbb{G} = \{G_1, G_2, ..., G_n\}$ by their resource requirements, where each job group $G_j=\{J_i | f(J_i) = S_j, \forall J_i \in \mathbb{J}\}_{i=1}^{m_j}$ contains all jobs with the same resource requirement. 
\name addresses the problem using a two-step approach, with each step occurring at a different level of scheduling granularity, as outlined in Algorithm~\ref{algo:amg-sched}: 
(i) Ordering \textbf{within a job group} to optimize local resource scheduling (\S\ref{sec:within}).
(ii) Merging orders \textbf{across job groups} to ensure global scheduling efficiency (\S\ref{sec:across}).
\name invokes Algorithm~\ref{algo:amg-sched} on job's request arrival and completion.
By breaking down the overall problem into two steps, we further reduce the problem's complexity without affecting the scheduling efficiency.
We provide theoretical insights to show the effectiveness of this problem decomposition in Appendix~\ref{sec:irs-lm}.

\subsubsection{Intra-Job Group Scheduling} 
\label{sec:within} 
  
Within each job group that shares the same device specifications, \name prioritizes jobs based on their remaining resource demand (Algorithm~\ref{algo:amg-sched} line~\ref{code:sort}). 
This ordering strategy aims to minimize the intra-group scheduling delay. 
Prioritizing jobs with smaller remaining resource demands has been shown to be effective in similar scheduling problems~\cite{srsf}. 
We choose this locally optimal scheduling strategy with the observation that it aligns well with the goal of achieving a globally optimal scheduling order.  
By default, the remaining resource demand refers to the needs of a single request within one round.
However, it can also encompass the total remaining demand for all upcoming rounds, provided such data is available.

\subsubsection{Inter-Job Group Scheduling} 
\label{sec:across} 
Addressing the scheduling problem across multiple job groups introduces an additional layer of complexity due to the intricate patterns of resource contention. 
Traditional scheduling algorithms, such as Random Matching and Shortest Remaining Service First (SRSF), are not designed to account for resource contention across job groups, leading to poor average JCT.

An effective scheduler should recognize and adapt to the resource contention patterns among jobs.
\name prioritizes groups with scarce resources (fewer eligible devices) to prevent delays from resource-rich groups.
Additionally, when a particular resource type is in high demand, the scheduler should judiciously allocate intersected resources to the job group with a longer queue, as this group contributes more significantly to the average scheduling delay.

To achieve this, \name allocates the intersected resources across different job groups based on two factors related to average scheduling delay: 
\begin{enumerate}
  \item \emph{Amount of eligible resources allocated}: the job group with smaller amount of eligible resources may have a longer scheduling delay under the same condition. 
  \item \emph{Queue length}: the job group with a longer queue length contributes more to the average scheduling delay as more jobs are waiting for the same type of resources. 
\end{enumerate}

Algorithm~\ref{algo:amg-sched} outlines the steps \name takes to allocate the current resource across job groups.
First, \name initializes resource allocation among job groups by starting to allocate resources to the job group with the most scarce resources (line~\ref{code:across}).
This results in an initial allocation plan with no resource sharing across job groups (lines~\ref{code:init-start}--\ref{code:init-end}), setting the stage for subsequent cross-group allocations.
  
To determine how to allocate intersected resources across job groups, \name greedily evaluates whether a job group with more abundant resources should acquire intersected resources from groups demanding more scarce resources with the objective of minimizing the average scheduling delay.
This evaluation starts with the job group possessing the most abundant resources (line~\ref{code:check}). 
If the resources allocated to this group remain unclaimed by other groups (line~\ref{code:not-none}), \name will decide how much resources from subsequent job groups should be allocated to it. 
Specifically, \name prioritizes job groups with longer queue lengths and fewer allocated resources, guided by a ratio that balances the number of affected jobs against the amount of allocated resources (line~\ref{code:condition}). 
If this ratio $\frac{m_j'}{|S_j'|}$ is larger than the one for the target resource group $\frac{m_k'}{|S_k'|}$, \name reallocates resources accordingly (lines~\ref{code:update-start}--\ref{code:update-end}). 
Otherwise, the algorithm ceases to allocate additional resources to the current job group from remaining groups (line~\ref{code:break}).
The reason is that if this job group needs more resources, it should first take the resources from job groups with relatively abundant resources. 

Function \texttt{get-queue-len}()  (line~\ref{code:queue-len}) would return the number of jobs $m'$ whose JCT would be delayed by potential group prioritization. 
For example, the affected queue length $m'$ may contain jobs from other job groups that have been deprioritized previously.
An easier way to approximate $m'$ is to use the group queue length.
If intersected resources have been decided to be allocated to job group $G_j$ over $G_k$, \name accumulates and updates their current resource allocation $S_j', S_k'$ and queuing length $m_j'$ (line~\ref{code:update-start}--\ref{code:update-end}). 

The \textbf{time complexity} of Algorithm~\ref{algo:amg-sched} is $\max (O(m\log m)$, $O(n^2))$, where $m$ is the number of ongoing jobs and $n$ is the number of job groups. 
We illustrate the \textbf{theoretical insight} behind the scheduling algorithm in Appendix~\ref{sec:proof} to illustrate the effectiveness of \name's approach.

\subsection{Device Matching}
\label{sec:match}

Now we focus on minimizing the response collection time, a significant contributor to the overall JCT, particularly when resource contention is low.
Existing cluster-level device-to-job matching solutions, either stick to a certain job order such as FIFO and SRSF, or match devices without a strategic algorithm such as random match, where none of them optimizes the job response collection time.

Response collection time is usually determined by the last successfully responding devices.
Hence, it can be reduced by allocating devices with similar higher capacity to the \fl job. 
Meanwhile, these high-end devices have lower probability to fail due to their quick task execution.
       
However, as mentioned in Section~\ref{sec:prob}, there is a trade-off between scheduling delay and response collection time. 
Intuitively, with limited device influx, priority should be given to minimizing the scheduling delay, which dominates the average JCT.
On the other hand, with sufficient device influx to fulfill a job request within a short period, we should consider minimizing the response collection time while obeying the scheduling order given by Section~\ref{sec:irs}.

\begin{algorithm}[t!]
  \caption{Device Matching}
  \label{algo:amg-match}
  \begin{algorithmic}[1]
   \STATE {\bfseries Input:} Jobs $J_{i}$, Devices $S_j' \in \name-SCHED(\mathbb{G}, \mathbb{S})$
   \STATE $S = \{ S^1, S^2, ..., S^V\}$ {$\triangleright$ \comm{Evenly partition devices}}
   \STATE $g_v = \frac{t^v}{t^0}, \forall v \in [1,V]$ $\triangleright$ \comm{Response time speed-up} \label{code:speedup}
   \STATE \textbf{Function:} \name-Match(Job $J_{i}$, Resource $S_j'$)
   \STATE $c_i = \frac{t_{response}}{t_{schedule}}$ $\triangleright$ \comm{Assign tiers in a rotating manner} \label{code:ratio}
\\
   $\triangleright$ \comm{Decide whether to trigger tier-based matching} \label{code:random}
   \STATE $u = \texttt{randint}(0, V)$ $\triangleright$ 
   \IF{$V + g_{u} \times c_i < c_i + 1$}  \label{code:jct}
    \STATE Update $S_{j}' = S_j' \cap S^{u}$ $\triangleright$ \comm{Assign tier $u$ devices to $J_{i}$} \label{code:tier}
   \ENDIF
   \STATE \textbf{Return:} $\{J_i, S_j'\}$
  \end{algorithmic}
\end{algorithm} 
To this end, we propose a resource-aware tier-based device-to-job matching solution to reduce the response collection time for each job as illustrated in Algorithm~\ref{algo:amg-match}. 

The matching algorithm is activated only for jobs that are currently served, as scheduled by Algorithm~\ref{algo:amg-sched}.
For each such job, \name partitions the eligible devices into $V$ tiers based on their hardware capabilities, where $V$ denotes the granularity of this partitioning.
If a job has been served before, \name adaptively sets the tier partition thresholds based on the hardware capacity distribution of the devices that participated in earlier rounds.
Otherwise, \name forgoes tier-based matching and profiles the devices allocated to the job's current request to inform future device tier partitioning.

For each served job request, \name randomly selects a device tier, denoted as $S^u$, to the job (Algorithm~\ref{algo:amg-match} line~\ref{code:random}).
This randomized tier assignment aims to expose each \fl job to a diverse set of devices, rather than confining them to high-end devices. 
Given that the response collection time is determined by the slowest responding participant, tier-based assignment does not adversely affect this metric.

As illustrated in Figure~\ref{fig:tier}, while tier-based matching may increase the scheduling delay by a factor of $V\geq1$, it can concurrently reduce the response collection time by a factor of $g\leq1$. 
The algorithm proceeds to perform such tier-based device-to-job matching for the job $J_i$ only if its JCT can be reduced, \ie, $1+c_i > V +  c_i g_u$ (line~\ref{code:jct}).
If the condition holds, \name allocates device tier $S_{u}$ to the job, effectively updating the set to $S_j' \cap S^{u}$.
Meanwhile the leftover device tiers would be allocated to subsequent jobs in the job group, maximizing the utilization of available resources.

To determine the response time speed-up factor $g$ for tier-based matching, we note that the device response time distribution adheres to a log-normal distribution~\citep{flint}. 
We use the 95th percentile as the statistical tail latency to account for the overall response collection time, thereby excluding failures and stragglers. 
\name profiles and estimates the response collection time for each device tier $v \in [1,..,V]$ and subsequently computes the speed-up factor $g_v = \frac{t^v}{t^0}$ relative to a non-tiered scenario (line~\ref{code:speedup}).

\begin{figure}[t!]
    \centering
    \includegraphics[trim=0 160 0 0,clip,scale=1]{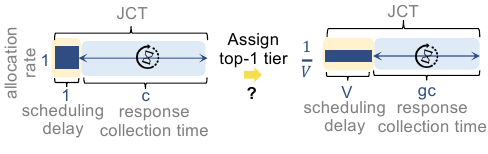} 
      \caption{Visualized tier-based device-to-job matching condition.}
       \label{fig:tier}
 \end{figure} 
 
\subsection{Enhancements}
\label{sec:enhancement}

\paragraph{Dynamic resource supply.}

As shown in Figure~\ref{fig:avail}, the total available \fl resources change significantly over time. 
To address this, \name continuously records each device eligibility through a time-series database. 
This database is then queried for resource eligibility distribution from the past time window.
However, relying solely on momentary eligible resource rates for input into the scheduling algorithm is inaccurate.
This is primarily due to the varying resource arrival patterns, as demonstrated in Figure~\ref{fig:avail}. 
Since \fl jobs span multiple days with diurnal resource patterns, \name averages eligibility over 24 hours for robust scheduling.
As a result, the scheduler can become both farsighted and robust, effectively accommodating the dynamic nature of resource availability. 

\paragraph{Starvation prevention.}  
Our heuristic can lead to larger \fl jobs being starved due to the preference given to smaller jobs. 
This is not acceptable especially when the jobs are initiated by different \fl developers who require performance guarantees. 
\name grants fairness to jobs to avoid such starvation. Specifically, our goal is to guarantee that the scheduling latency of a job $J_i$ is no worse than fair sharing, which is defined as $T_i=M*sd_i$, where $M$ is the number of simultaneous \fl jobs and $sd_i$ represents the JCT without contention. 
Then, we adjust each job demand to be $d_i' = d_i \times  (\frac{t_i}{T_i})^{\epsilon}$ to ensure fairness within a job group, and adjust each group queue length $q_j' = q_j  \times  (\frac{\sum_{J_i \in G_j} T_i}{ \sum_{J_i \in G_j} t_i})^{\epsilon}$ to ensure fairness across job groups.
$t_i$ is the time usage of job $J_i$ at the moment and $\epsilon \in [0, \infty )$ is a fairness control knob. 
When $\epsilon=0$, the algorithm is identical to the one in Section~\ref{sec:irs}. 
As $\epsilon \rightarrow \infty$, the fairness multiplier dominates the scheduling, resulting in maximum fairness. We show that \name improves JCT over its counterparts with our starvation design (\S\ref{sec:ablation}).

%% file: pages/eval.tex
\section{Evaluation}
\label{sec:eval} 
In this section, we evaluate the effectiveness of \name through event-driven simulation and testbed experiments. 
 Our key takeaways are:
\begin{itemize}
\item \name reduces average job completion time (JCT) by up to $1.88\times$ compared to state-of-the-art baselines, without compromising model accuracy, across diverse real-world \fl workloads (\S\ref{sec:macro}).
\item \name outperforms its counterparts by intelligently scheduling jobs and matching devices to jobs, leveraging different design components (\S\ref{sec:breakdown}).
\item \name's benefits are robust under a wide range of \fl workloads and environment setups (\S\ref{sec:ablation}).
\end{itemize}

\subsection{Experiment Setup}
\label{sec:setup}

\begin{figure}[t!]
   \centering 
    \begin{subfigure}[t]{0.24\textwidth}
      \includegraphics[trim=0 0 0 0,clip,scale=0.3]{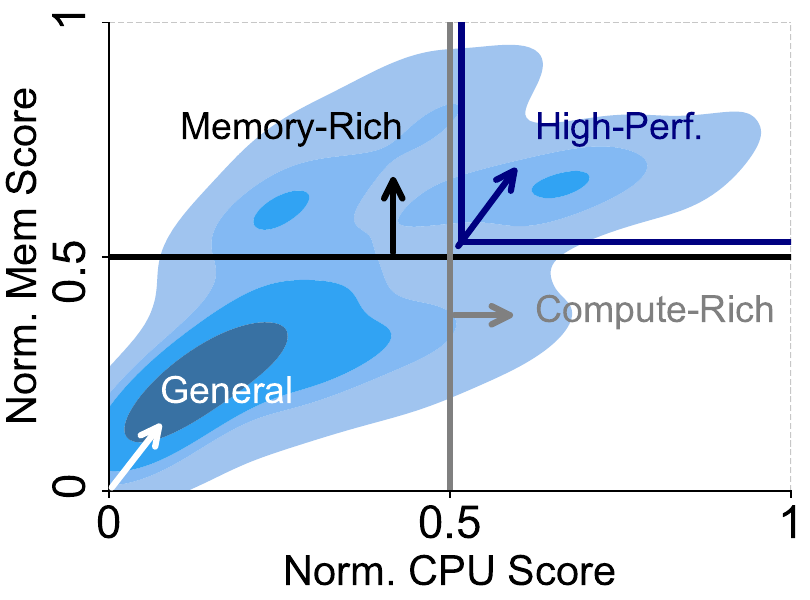} 
     \caption{Device eligibility trace.} 
   \label{fig:elig-trace} 
     \end{subfigure}
         \begin{subfigure}[t]{0.23\textwidth}
      \includegraphics[trim=0 0 0 0,clip,scale=0.3]{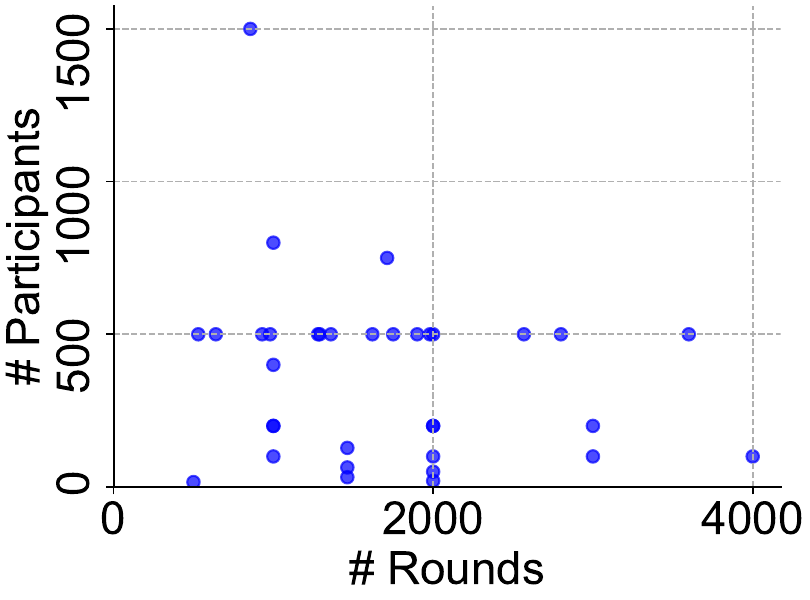} 
     \caption{\fl job demand trace.} 
   \label{fig:job-trace} 
     \end{subfigure}
  \caption{Device and job trace used in experiments. (a) Devices are stratified into four regions  to explore different overlap patterns. (b) The diverse workloads in experiments are derived from the job demand trace based on demand characteristics. }
\end{figure}

\paragraph{Testbed.} 
To rigorously evaluate \name, we employ a two-pronged approach. 
First, we have developed a high-fidelity simulator that replays client and job traces, effectively emulating the dynamics of the scheduling environment for large-scale evaluations.
Second, we deploy real \fl systems to execute actual \fl jobs at a smaller scale of devices.

\paragraph{Resources.}
To faithfully emulate heterogeneous device runtimes, networking, and availability, we use device traces from FedScale~\cite{fedscale} (Figure~\ref{fig:client}) and AI Benchmark~\cite{ai-bench} (Figure~\ref{fig:elig-trace}). 
%These traces enable us to model diverse computation and communication latencies, informed by device capabilities and model complexities. 
Each unique device trace is limited to one \fl job per day for realism.

\paragraph{\fl jobs.}  
We evaluate \name over multiple synchronous \fl jobs~\cite{google-query-sugg}, where each successful training round requires a minimum of 80\% target participants to report back within a deadline, which is set to be 5min - 15min depending on the round demand.
Our approach also extends to asynchronous \fl jobs, as scheduling decisions depend solely on remaining resource demand, not request submission timing.
Jobs are drawn from diverse real-world \fl applications~\cite{gboard-fl, google-query-sugg, google-keyboard, google-emoji, google-speech, google-health}, whose resource demand is depicted in Figure~\ref{fig:job-trace}. 
In the real \fl experiment, each job aims to train a ResNet-18~\cite{resnet} and MobileNet-V2~\cite{mobilenet} on FEMNIST dataset.

\paragraph{Workloads.} 
Our evaluation includes five workload scenarios that sample differently from the job trace in Figure~\ref{fig:job-trace} to rigorously evaluate \name's performance.
\emph{Even:} Sampled from all jobs, which is the default trace.
\emph{Small:} Uniformly sampled only from jobs with below-average total demand.
\emph{Large:} Uniformly sampled only from jobs with above-average total demand.
\emph{Low:} Uniformly sampled only from jobs with below-average demand per round.
\emph{High:} Uniformly sampled only from jobs with above-average demand per round.
Default simulation and real-world workloads contain 50 and 20 jobs, respectively.
Jobs arrive via a Poisson process with a 30-min average inter-arrival.

Device requirements are stratified into four categories based on the CPU and memory capacities (Figure~\ref{fig:elig-trace}) to create various resource contention patterns where the eligible resources for each job may overlap, contain, or be within the eligible resources of other jobs.
By default, each job is randomly mapped to one category among \emph{General} resources, \emph{Compute-Rich} resources, \emph{Memory-Rich} resources, \emph{High-Performance} resources.

\paragraph{Baselines.}
We compare \name against FIFO, SRSF, and an optimized random matching. 
The random matching algorithm typically assigns devices to eligible jobs randomly but is optimized here to schedule jobs in a randomized order, reducing round abortions under contention and establishing a stronger baseline.
Note that we only run \name with the starvation prevention strategy in Section~\ref{sec:ablation}.

\paragraph{Metrics.} 
Our primary performance metrics include the average job completion time (JCT).  
Note that while \name does not explicitly optimize for \fl job accuracy, it does not adversely affect it either.

\definecolor{lightgray}{gray}{0.95}
\newcolumntype{v}{>{\columncolor{lightgray}}c}
\begin{table}[t]
\small
\centering
\begin{tabular}{lccv}
%\hline
\Xhline{3\arrayrulewidth}
%\rowcolor[HTML]%{ECF4FF} 
      & \textbf{FIFO}  & \textbf{SRSF}  & \textbf{\name}            \\ \hline
%\rowcolor[HTML]{FFFFFF} 
\textbf{Even}  & 1.38$\times$ & 1.69$\times$ & 1.87$\times$ \\ %\hline
%\rowcolor[HTML]{ECF4FF} 
\textbf{Small} & 1.48$\times$ & 1.68$\times$ & 1.78$\times$ \\ %\hline
%\rowcolor[HTML]{FFFFFF} 
\textbf{Large} & 1.64$\times$ & 1.57$\times$ & 1.72$\times$ \\ %\hline
%\rowcolor[HTML]{ECF4FF} 
\textbf{Low}   & 1.55$\times$ & 1.66$\times$ & 1.88$\times$ \\ %\hline
%\rowcolor[HTML]{FFFFFF} 
\textbf{High}  & 1.42$\times$ & 1.41$\times$ & 1.63$\times$ \\ \Xhline{3\arrayrulewidth}
\end{tabular} 
\caption{Summary of improvements on average JCT over random matching on different \fl workloads.}
\label{tbl:res} 
\end{table}

%\begin{figure}[t!]
%   \centering 
%    \begin{subfigure}[t]{0.24\textwidth}
%      \includegraphics[trim=0 0 0 0,clip,scale=0.3]{jct-cdf-tiny.pdf} 
%     \caption{Low workload.} 
%   \label{fig:jct-cdf-tiny} 
%     \end{subfigure}
%         \begin{subfigure}[t]{0.23\textwidth}
%      \includegraphics[trim=0 0 0 0,clip,scale=0.3]{jct-cdf-huge.pdf} 
%     \caption{High workload.} 
%   \label{fig:jct-cdf-huge} 
%     \end{subfigure}
%  \caption{ \name achieve a higher ratio of jobs that complete faster. } 
%   \label{fig:jct-cdf} 
%\end{figure}

\subsection{End-to-End Performance }
\label{sec:macro}  

\textbf{\name achieves better average JCT Improvement.}
We assess the performance of different scheduling algorithms by evaluating its performance over different workloads.
We report the average JCT speed-up for each scheduling algorithm compared to the random scheduling in Table~\ref{tbl:res}. 
We observe that our scheduling algorithm consistently provides stable improvements in the average JCT across various workloads, which underscores the robustness of \name.
%We discuss the potential improvement when the total number of rounds is known a priori in Appendix~\ref{sec:clairvoyant}.

%\textbf{\name achieves higher ratio of jobs complete faster.}
%We report the cumulative distribution function (CDF) of JCT in Figure~\ref{fig:jct-cdf} for each scheduling algorithm.
%We observe that \name achieves a higher ratio of jobs that complete faster compared to the other algorithms. 

\textbf{\name achieves faster convergence without affecting accuracy.}
We report the final model test accuracy of \fl jobs under different schedules with the help of our \fl system at a smaller experiment scale.
As shown in Figure~\ref{fig:acc-time}, we observe that \name does not affect the final model test accuracy but speeds up the overall convergence process. 
%  improves the convergence speed by serving job resource requests as soon as possible, while not affect the final job accuracy.
%more discussion about the convergence performance.

\textbf{\name has negligible overhead.}
We emulated a large number of \fl jobs and groups to evaluate the scheduler's scalability. 
Our results in Figure~\ref{fig:algo-latency} demonstrate that the latency incurred by one-time triggering for scheduling and matching remains low, even with a substantial increase in job and group numbers.
This is due to its low time complexity of $\max(O(m\log m), O(n^2))$, where $m$ is the number of jobs and $n$ is the number of device groups.

\begin{figure}[t!]
   \centering 
     \begin{minipage}{1.57in}
      \includegraphics[trim=0 0 0 0,clip,scale=0.3]{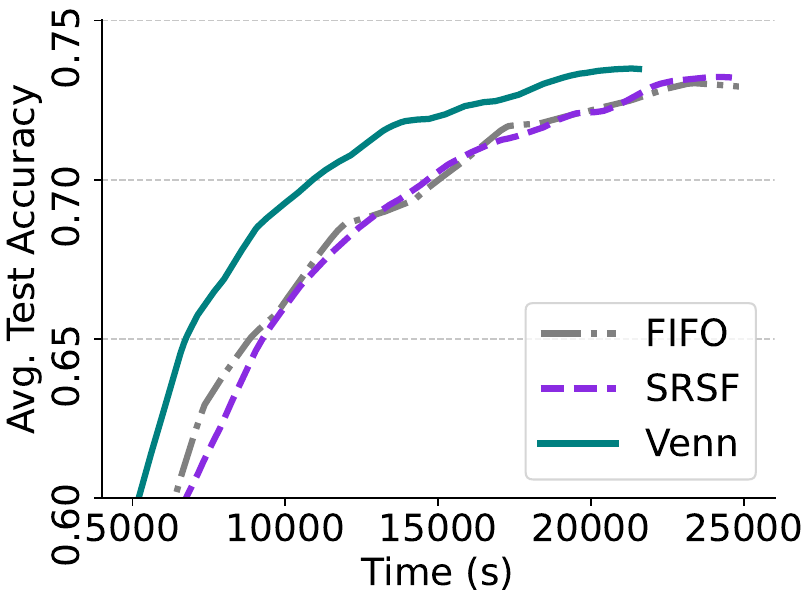} 
     \caption{\name does not affect the average test accuracy. }
   \label{fig:acc-time}
     \end{minipage}
     \begin{minipage}{1.48in} 
      \includegraphics[trim=0 0 0 0,clip,scale=0.3]{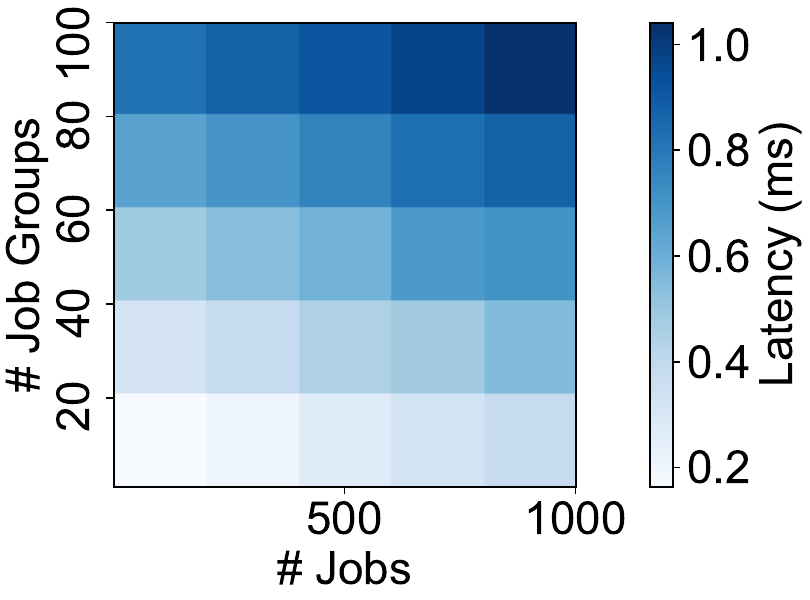} 
     \caption{\name introduces negligible overhead at scale.} 
   \label{fig:algo-latency} 
     \end{minipage}
\end{figure}

\subsection{Performance Breakdown}
\label{sec:breakdown}

\begin{figure}[t!]
   \centering 
    \begin{subfigure}[t]{0.23\textwidth} 
      \includegraphics[trim=0 0 0 0,clip,scale=0.3]{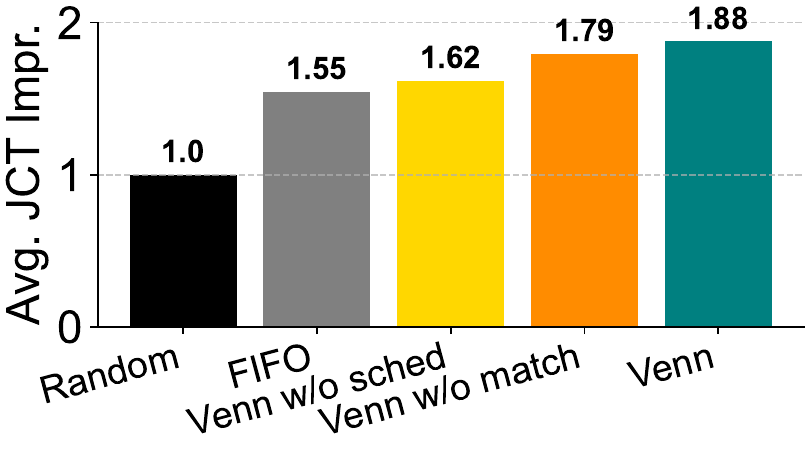} 
     \caption{Low workload.} 
   \label{fig:breakdown-low} 
     \end{subfigure}
         \begin{subfigure}[t]{0.23\textwidth} 
      \includegraphics[trim=0 0 0 0,clip,scale=0.3]{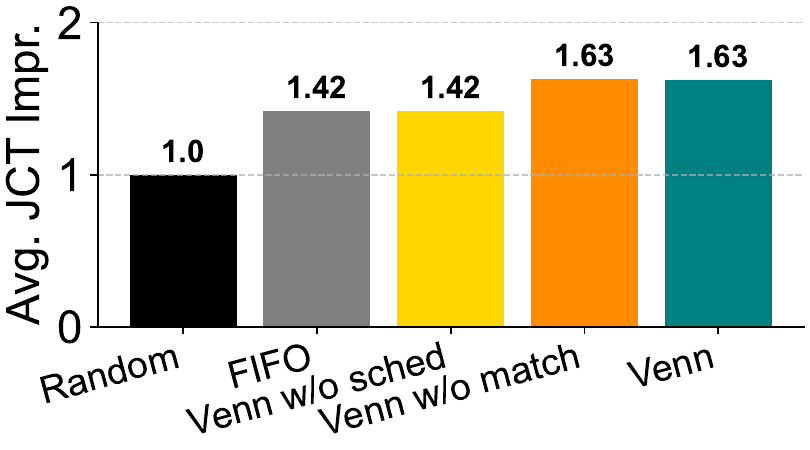} 
     \caption{High workload.} 
   \label{fig:breakdown-high} 
     \end{subfigure}
          \caption{Average JCT improvement breakdown.} 
   \label{fig:breakdown} 
\end{figure}

\begin{table}[t!]
\centering
\begin{tabular}
{ 
@ {}
>{\columncolor[HTML]{FFFFFF}}c@{\hskip 0.08in} |
>{\columncolor[HTML]{FFFFFF}}r@{\hskip 0.1in}
>{\columncolor[HTML]{EFEFEF}}r@{\hskip 0.1in}
>{\columncolor[HTML]{FFFFFF}}r@{\hskip 0.1in}}
\Xhline{3\arrayrulewidth}
%& \multicolumn{3}{c|}{\cellcolor[HTML]{FFFFFF}{\color[HTML]{000000} Job Demand}} \\ \hline
& \cellcolor[HTML]{FFFFFF}{\color[HTML]{000000} \textbf{25th}} & \cellcolor[HTML]{FFFFFF}{\color[HTML]{000000} \textbf{50th}} & \cellcolor[HTML]{FFFFFF} {\color[HTML]{000000} \textbf{75th}}\\ \hline
{\color[HTML]{000000} \textbf{Even}} & {\color[HTML]{000000} 11.5$\times$} & {\color[HTML]{000000} 7.2$\times$} & {\color[HTML]{000000} 5.6$\times$} \\ \cline{1-1}
{\color[HTML]{000000} \textbf{Small}} & {\color[HTML]{000000} 6.8$\times$} & {\color[HTML]{000000} 5.2$\times$} & {\color[HTML]{000000} 4.3$\times$} \\ \cline{1-1}
{\color[HTML]{000000} \textbf{Large}} & {\color[HTML]{000000} 3.7$\times$} & {\color[HTML]{000000} 2.9$\times$} & {\color[HTML]{000000} 2.7$\times$} \\ \cline{1-1}
{\color[HTML]{000000} \textbf{Low}} & {\color[HTML]{000000} 11.6$\times$} & {\color[HTML]{000000} 7.5$\times$} & {\color[HTML]{000000} 4.7$\times$} \\ \cline{1-1}
{\color[HTML]{000000} \textbf{High}} & {\color[HTML]{000000} 5.1$\times$} & {\color[HTML]{000000} 3.3$\times$} & {\color[HTML]{000000} 3.1$\times$} \\ 
\Xhline{3\arrayrulewidth}
\end{tabular}
\caption{Breakdown of average JCT improvement across jobs with lowest 25\%, 50\%, and 75\% of total demands.
Venn benefits more on smaller jobs.}
\label{tbl:job-breakdown-part1}
\end{table}

\begin{table}[t!]
\centering
\begin{tabular}
{ 
@ {}
>{\columncolor[HTML]{FFFFFF}}c@{\hskip 0.08in} |
>{\columncolor[HTML]{FFFFFF}}c@{\hskip 0.1in}
>{\columncolor[HTML]{EFEFEF}}c@{\hskip 0.1in}
>{\columncolor[HTML]{FFFFFF}}c@{\hskip 0.1in}
>{\columncolor[HTML]{EFEFEF}}c@{\hskip 0.1in} }
\Xhline{3\arrayrulewidth}
%\multicolumn{1}{l|}{\cellcolor[HTML]{FFFFFF}{\color[HTML]{000000} }} & \multicolumn{4}{c}{\cellcolor[HTML]{FFFFFF}{\color[HTML]{000000}  Res. Requirement Type}}                                                                                                                                                                                                                                           \\ \hline
\multicolumn{1}{l|}{\cellcolor[HTML]{FFFFFF}{\color[HTML]{000000} }} & \multicolumn{1}{c}{\cellcolor[HTML]{FFFFFF}{\color[HTML]{000000} \textbf{General.}}} & \multicolumn{1}{c}{\cellcolor[HTML]{FFFFFF}{\color[HTML]{000000} \textbf{Compute.}}} & \multicolumn{1}{c}{\cellcolor[HTML]{FFFFFF}{\color[HTML]{000000} \textbf{Memory.}}} & \multicolumn{1}{c}{\cellcolor[HTML]{FFFFFF}{\color[HTML]{000000} \textbf{High-perf.}}} \\ \hline
{\color[HTML]{000000} \textbf{Even}}                                 & {\color[HTML]{000000} 1.5$\times$}                                                 & {\color[HTML]{000000} 7.2$\times$}                                                 & {\color[HTML]{000000} 5.3$\times$}                                                 & {\color[HTML]{000000} 3.9$\times$}                                                 \\ \cline{1-1}
{\color[HTML]{000000} \textbf{Small}}                                & {\color[HTML]{000000} 0.9$\times$}                                                 & {\color[HTML]{000000} 6.0$\times$}                                                 & {\color[HTML]{000000} 2.8$\times$}                                                 & {\color[HTML]{000000} 2.6$\times$}                                                 \\ \cline{1-1}
{\color[HTML]{000000} \textbf{Large}}                                & {\color[HTML]{000000} 0.9$\times$}                                                 & {\color[HTML]{000000} 3.7$\times$}                                                 & {\color[HTML]{000000} 1.8$\times$}                                                 & {\color[HTML]{000000} 2.6$\times$}                                                 \\ \cline{1-1}
{\color[HTML]{000000} \textbf{Low}}                                  & {\color[HTML]{000000} 0.8$\times$}                                                 & {\color[HTML]{000000} 3.4$\times$}                                                 & {\color[HTML]{000000} 2.1$\times$}                                                 & {\color[HTML]{000000} 8.7$\times$}                                                 \\ \cline{1-1}
{\color[HTML]{000000} \textbf{High}}                                 & {\color[HTML]{000000} 0.8$\times$}                                                 & {\color[HTML]{000000} 2.2$\times$}                                                 & {\color[HTML]{000000} 2.2$\times$}                                                 & {\color[HTML]{000000} 5.6$\times$}                                                 \\ 
\Xhline{3\arrayrulewidth}
\end{tabular}
\caption{Breakdown of average JCT improvement across jobs that ask for General resources, Compute-rich resources, Memory-rich resources and High-performance resources. 
\name benefits more on jobs that ask for scarcer resources.}
\label{tbl:job-breakdown-part2} 
\end{table}

\begin{table}[t]
  \small
  \centering
  \begin{tabular}{lccv}
  %\hline
  \Xhline{3\arrayrulewidth}
  %\rowcolor[HTML]%{ECF4FF} 
        & \textbf{FIFO}  & \textbf{SRSF}  & \textbf{\name}            \\ \hline
  \textbf{General}  & 1.46$\times$ & 1.78$\times$ & 1.94$\times$ \\ %\hline
  \textbf{Compute-heavy} & 1.73$\times$ & 2.08$\times$ & 2.23$\times$ \\ %\hline
  \textbf{Memory-heavy} & 1.68$\times$ & 2.05$\times$ & 2.27$\times$ \\ %\hline
  \textbf{Resource-heavy}  & 1.65$\times$ & 1.90$\times$ & 2.01$\times$ \\ \Xhline{3\arrayrulewidth}
  \end{tabular} 
  \caption{Average JCT improvement on four biased workloads. }
  % and breakdown performance for jobs with different resource requirements. }
  %Venn benefits more on workloads with biased resource requirements. 
  \label{tbl:bias} 
  \end{table}

We present a performance breakdown of \name, which consists of two parts: a job scheduling algorithm that determines the job order to minimize the scheduling delay, and a device-to-job matching algorithm to reduce the response collection time.   
We evaluate the performance of \name with only the scheduling algorithm (\name w/o matching), \name with only matching algorithm and FIFO (\name w/o scheduling), and \name with both algorithms (\name). 
We show the improvement of the average JCT over the default random scheduling for each component. 
As shown in Figure~\ref{fig:breakdown}, the tier-based device-to-job matching algorithm primarily benefits low workload where the resource contention is small, which is aligned with our original design intention.
The reason is that when resource supply is sufficient, the response collection time would dominate the JCT, which can be optimized by our matching algorithm.

To analyze the impact of \name on different types of jobs, we break down jobs by their total demands and device requirements (Figure~\ref{fig:elig-trace}), and then analyze the average JCT improvement for each type. 
Table~\ref{tbl:job-breakdown-part1} and Table~\ref{tbl:job-breakdown-part2} quantifies how \name improves average JCT across varying total demands (25th, 50th, 75th percentiles) and eligibility types (General, Compute-rich, Memory rich, High-performance ). 
Notably, jobs with smaller total demands and scarcer resources benefit the most from \name.

\subsection{Case Study on Biased Workload}
\label{sec:biased}
This section delves into an in-depth analysis of \name's adaptability and performance across four distinct workloads, each characterized by a specific bias in job resource requirements. 
These workloads include General, Compute-Heavy, Memory-Heavy, and Resource-Heavy categories. 
For example, the Compute-Heavy workload is structured such that half of its jobs are predominantly geared towards compute-intensive resources, with the rest evenly distributed across the other three resource types. 
This setup introduces varied queue lengths in different job groups, providing a robust testbed for evaluating \name's capability in effectively managing these variations.

The design of these workloads aims to scrutinize \name's proficiency in navigating diverse resource requirement distributions, while maintaining uniformity in job demands as illustrated in Figure~\ref{fig:job-trace}. The outcomes of these experiments are systematically presented in Table~\ref{tbl:bias}, offering insights into the algorithm's performance under each workload scenario.

\subsection{Ablation Study}
\label{sec:ablation}

\begin{figure}[t!]
   \centering 
     \begin{minipage}{1.5in}
      \includegraphics[trim=0 0 0 0,clip,scale=0.3]{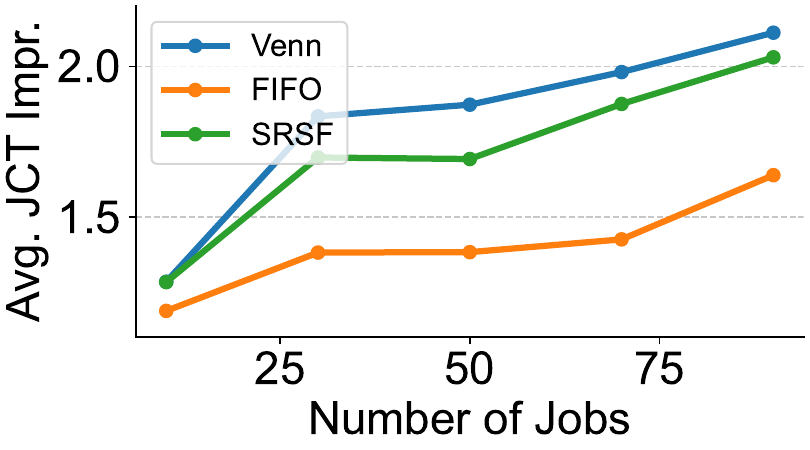} 
     \caption{\name outperforms FIFO and SRSF across different numbers of jobs.} 
   \label{fig:num-job} 
        \end{minipage} 
      \hfill
     \begin{minipage}{1.5in}
      \includegraphics[trim=0 0 0 0,clip,scale=0.3]{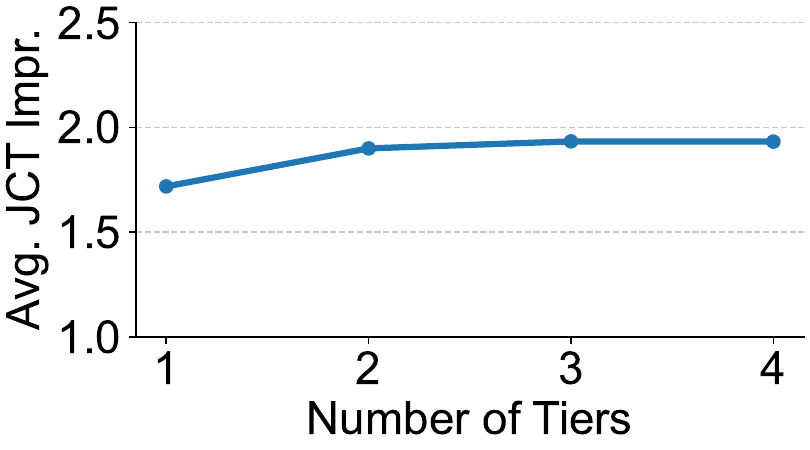} 
     \caption{\name's improvement across different numbers of tiers.} 
   \label{fig:num-tier} 
        \end{minipage} 
\end{figure}
  
\paragraph{Impact of number of jobs.}
We evaluate \name with different numbers of \fl jobs arriving over time. 
As the number of jobs increases, resource contention becomes more pronounced, highlighting the importance of efficient scheduling under such conditions. 
We present the average JCT speed-up with different numbers of jobs in even workload to demonstrate the effectiveness of our algorithm. 
As shown in Figure~\ref{fig:num-job}, \name consistently provides improvement across various numbers of jobs, with its benefits becoming more pronounced as the number of jobs increases.

\paragraph{Impact of number of tiers.}
We evaluate the matching algorithm's performance across varying numbers of client tiers, ranging from a single tier to multiple. 
Figure~\ref{fig:num-tier} shows that increased tier granularity enhances device-to-job matching and improves performance. However, the gains plateau beyond a certain point, as finer tiers increase scheduling delay without yielding further reductions in response collection time. Thus, optimizing the number of tiers is crucial for balancing scheduling efficiency and JCT improvement.

\begin{figure}[t!]
   \centering 
    \begin{subfigure}[t]{0.23\textwidth} 
      \includegraphics[trim=0 0 0 0,clip,scale=0.3]{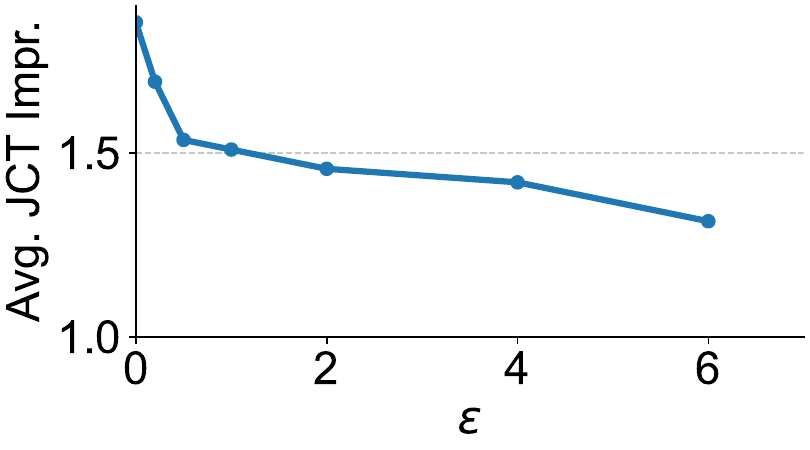} 
     \caption{\name's improvement over different $\epsilon$.} 

\label{fig:fairness} 
     \end{subfigure}
         \begin{subfigure}[t]{0.23\textwidth} 
      \includegraphics[trim=0 0 0 0,clip,scale=0.3]{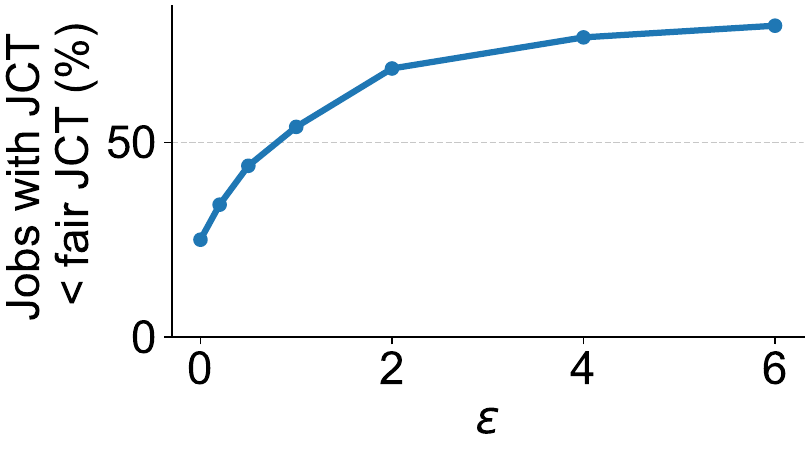} 

     \caption{Ratio of jobs meet the fair-share JCT.} 
  \label{fig:ddl} 
     \end{subfigure}

          \caption{Fairness knob.} 
\end{figure}

\paragraph{Impact of fairness knob.}
We incorporated a fairness knob ($\epsilon$) to strike a balance between performance and fairness. 
We tune the value of $\epsilon$ and report the average JCT speed-up against these values in Figure~\ref{fig:fairness}. The results demonstrate that, as $\epsilon$ increases, the JCT speed-up tends to decrease. 
As shown in Figure~\ref{fig:ddl}, the percentage of jobs that meet the fair-share JCT increases with the $\epsilon$, where  $\epsilon=$2 gives 69\% of jobs receive their fair-share JCT.
This observation highlights the trade-off between performance and fairness within our \fl resource scheduling algorithm, which can be fine-tuned by adjusting the value of $\epsilon$.

%% file: pages/related.tex
\section{Related Work}

\textbf{Cluster resource managers.} 
There are many cluster resource managers that schedule resources with constraints~\cite{phoenix, choosy, gavel}.
Existing ML cluster resource managers mainly focus on managing the stable data-center resources like GPU and CPU~\citep{gandiva, synergy} in order to improve JCT, utilization and fairness~\citep{optimus, gandiva-fair}. 
Some research delves into GPU-specific optimizations~\cite{tiresias, fluid, afs}, while others co-design resource managers with the specific characteristics of ML workloads~\cite{pollux, zeus}. 
However, they are mostly designed for data center and fail to capture the level of availability and heterogeneity of \fl resources, nor do they consider both scheduling delay and response collection time of \fl jobs.

% Moreover, we want to note that \name is complementary to these in-cluster schedulers. For example, when the number of rounds of \fl jobs is not available, we can leverage Least-Attained Service~\cite{tiresias} to performed \name scheduling within each Tiresias-determined job queue. 
% However, \fl jobs can have different number of rounds, which is not known a priori.
% \name is designed to seamlessly integrate algorithms like Least-Attained Service~\cite{tiresias} to perform job-level scheduling without knowing the number of rounds each \fl job may submit

\textbf{\fl client selectors.} 
Several recent works have studied client selection at the single \fl job level. 
In addition to enforce device requirements including software version, hardware capacity and data quality, they further cherry-pick clients based on their state, system and statistical utility~\citep{oort,relay,divfl,power-of-choice,tifl,gluefl} to speed up the training. 
However, they solely focus on response collection time~\citep{multifl,auxo} and overlook the time required to acquire adequate resources. 
Additionally, optimizing individual \fl job performance is insufficient as the deployment scale of \fl applications continues to grow.

\textbf{\fl resource managers.} 
Large companies including Apple, Meta and Google have proposed their \fl infrastructures; however, \fl resource management is not their primary focus. 
These resource managers simply adopt random device-to-job matching in various forms, resulting in suboptimal scheduling delays and response collection times.

%Additionally, their resource managers predominantly focus on serving their own \fl jobs, featuring single request and aggregation strategies. 

%% file: pages/conclusion.tex
\section{Conclusion}

In this paper, we present \name, our \fl resource manager designed to efficiently share large-scale heterogeneous device resources among multiple \fl jobs with diverse requirements.
\name features a contention-aware job scheduling algorithm and a resource-aware device-to-job matching algorithm, aiming to minimize the average job completion time (JCT) for \fl workloads.
Our evaluation using a range of real-world \fl workloads demonstrates that \name improves average JCT by up to $1.88\times$.

%% file: pages/ack.tex
\section*{Acknowledgements}
We thank the anonymous reviewers for valuable feedback.
This work was supported in part by NSF grant CNS-2106184 and a grant from Cisco.

%% file: pages/appendix.tex
\section{Detailed \name Responsibility}
\label{sec:app-resp}

\name delegates responsibilities such as device selection, device fault tolerance, and privacy protection to individual \fl jobs. 
Device failures are both inevitable and difficult to predict in \fl. 
Rather than imposing a one-size-fits-all solution, \name empowers \fl jobs to take the reins on fault tolerance based on their specific workloads and objectives.
Therefore, \name offloads handling device fault tolerance to \fl jobs, who can better detect and react to device failures (\eg, deciding the amount of overcommit \cite{google}).
Similarly, \name offers \fl jobs the freedom to design their own device selectors~\cite{oort}, where they can incorporate customized resource criteria into their requests. 
\name also does not interfere with job-specific privacy solutions such as secure aggregation~\cite{sec-aggr,meta} or differential privacy~\cite{fl-dp,gboard-fl}.

 \section{ILP Formulation of IRS}
 \label{sec:ilp}

We now formulate the IRS that allocates resources to jobs under the constraints with the objective of minimizing the average scheduling delay.
Assume we have devices $\mathbb{S} = \{s_1, s_2, \allowbreak ..., s_q\}$ continuously arriving at times $\{t_i, t_2, \allowbreak ..., t_q\}$.
There are $m$ jobs $\mathbb{J}= \{J_1,J_2, \allowbreak ..., J_{m} \}$ with their resource demands $\mathbb{D}= \{D_1,D_2, \allowbreak ..., D_{m} \}$.
Let $e_{ij}$ be a binary variable in the eligibility matrix, which is set to 1 if device $i$ is eligible to job $j$, and 0 otherwise. 
Let $x_{ij}$ be a binary variable of resource allocation, which is 1 if device $i$ is assigned to job $j$, and 0 otherwise. 

We have to follow these constraints during the resource allocation:
$$
\sum_{j=1}^m x_{ij} \leq 1, \forall i \in [1, q]
$$
$$
\sum_{j=1}^m x_{ij} \times e_{ij} \leq 1, \forall i \in [1, q]
$$
$$
\sum_{i=0}^q x_{ij} = D_j, \forall j \in [1, m]
$$

Therefore, the scheduling delay of each job is determined by the time it acquires the last needed device, \ie, $T_j = \max_i(x_{ij} \times t_i )$ under these constraints. 
The overall objective can then be expressed as:
$$
\min \frac{\sum_{j=1}^{m} T_j}{m}
$$

 \section{Theoretical Insight to the Heuristic of IRS}
 \label{sec:irs-lm}

\newtheorem{lemma}{Lemma}
\begin{lemma}
Given a diverse set of \fl jobs with one round request, if jobs are scheduled optimally in terms of the average JCT, first within each job group and then across job groups, the resulting average JCT is optimal. 
\label{lm:two-step}
\end{lemma}

\begin{proof}

Let us assume there is an optimal scheduling algorithm that optimizes the average JCT within each group, and there is an optimal scheduling algorithm which decides how to merge the job order across job groups to minimize the average JCT.
Since the second step is assumed to provide optimal average JCT based on the previous within group job order,
we only need to prove the global optimal schedule follows the order generated by the within job group step.

\name employs smallest remaining job demand first algorithm within each job group.
Since prioritizing jobs with smaller remaining resource demands has been shown to be effective in similar scheduling problems~\cite{srsf}, we skip the proof that the scheduling algorithm within job group gives the local optimal average JCT for each group.

We prove the rest by contradiction. 
Assume that there exists an optimal schedule $S$ that does not follow the order given by each job group.
In this assumed optimal schedule $S$, let us say there are two jobs $J_A$  and $J_B$ in the same group such that $J_A$ comes after $J_B$, but $J_A$ has fewer resource requirements than $J_B$.
Let us swap $J_A$ and $J_B$ to create a new schedule $S'$. 
Since $J_A$ has fewer resource demand, the average JCT of $S'$ will be less than that in $S$.
This contradicts our original assumption that $S$ is an optimal schedule, as we've found a schedule $S'$ with a lower average JCT.
Therefore, the assumption is false, and the order given by each job group (sorted by resource demands) must be part of the optimal schedule.
If we have an optimal scheduling across job groups, the overall average JCT will be optimal.

\end{proof}

 \section{Effectiveness of \name Scheduling Heuristic}
 \label{sec:proof}
 
To illustrate the effectiveness of \name's approach, 
we start with proving Lemma~\ref{lemma:two-groups}, which considers a simplified case involving only two job groups with arbitrary resource contention patterns.
 Through mathematical proof, we can demonstrate that our algorithm achieves the optimal solution under this setting.

\begin{lemma}
Given two job groups with arbitrary resource contention patterns, the scheduling plan generated by \name as in Algorithm~\ref{algo:amg-sched} is capable of minimizing the average scheduling delay, if a future resource allocation plan is set.

\label{lemma:two-groups}
\end{lemma}

To better prove the Lemma, we introduce a new representation of the scheduling problem in a more scalable way.
Firstly, as depicted in Figure~\ref{fig:step-setup}, we represent the two job groups by two distinct sets of squares, where the area of each square corresponds to the size of the request demand for that job.

Secondly, to visualize the temporal dynamics of resource allocation, we refer to Figure~\ref{fig:amg-step2}. 
For the sake of this example, let's assume a constant inflow of 100 devices per time unit. 
Within this set, `x' devices possess memory  $\geq$ 2GB, while all 100 devices have memory  $\geq$ to 1GB. 
The y-axis is partitioned into two segments: the 0 to `x' range signifies devices with memory exceeding 2GB, and the `x' to 100 range represents devices with memory ranging between 1GB and 2GB.

Resource allocation over time is illustrated using rectangles, each indicating the job request to which devices are assigned. 
For instance, in the right subfigure of Figure~\ref{fig:amg-step2}, devices in the 0 to `x' memory range are allocated to job group B at time 0, while those in the `x' to 100 range are allocated to job group A. 
This representation allows us to dynamically track resource allocation across different jobs over time.

\begin{proof}

As shown in Figure~\ref{fig:step-setup}, there are $m_A$ requests that ask for devices with 1GB memory and $m_B$ jobs that request devices with 2GB memory, resulting in two job groups $A$ and $B$.  
The devices constantly check-in and execute one \fl task, where 100\% devices with memory size $\geq$ 1GB and x\% of the devices have memory size $\geq$ 2GB. 
Note that, the proof is not limited to the contention pattern draw in Figure~\ref{fig:step-setup}, it can be generalized to job group with intersected resource contention and give the same conclusion.

Based on Algorithm~\ref{algo:amg-sched},
the first step is to sort these jobs within each job group by job size in ascending order (Figure~\ref{fig:amg-step1}). 
In the second step, we generate an initial resource allocation for each job group by focusing on the job group with the scarcest resources.
This results in an initial allocation plan that avoids resource sharing across job groups, setting the stage for subsequent cross-group allocations.

Based on the group-level initial allocation plan (left subfigure in Figure~\ref{fig:amg-step2}), we need to determine the job order across groups, that is, to decide whether to prioritize jobs from Group $A$ over Group $B$ (right subfigure in Figure~\ref{fig:amg-step2}) at current time in order to achieve a smaller average scheduling delay.
In this case, we focus on determining the order of the first job with size $l$ in Group A and calculate the queuing delay difference ($\Delta t$) if we prioritize the first job from Group $A$ over Group $B$.
$$
\Delta t = l * m_B' - (\frac{l}{1-x}-l) * m_A'
$$
where $m_A'$, $m_B'$ represents the number of remaining jobs whose queuing delay may be affected by this prioritization. 
Since the future resource allocation is set by the previous initial allocation or assumed to be given, $m_A'$, $m_B'$ are feasible to get.
We prioritize the first job from Group A only if $\Delta t < 0 $, which gives $\frac{m_A'}{1-x} > \frac{m_B'}{x}$, otherwise we stick with the original plan. 
$\Delta t < 0 $ is actually the prototype of the scheduling decision as in Algorithm~\ref{algo:amg-sched} line~\ref{code:condition}.

\end{proof}

\begin{figure}[h!]
  \centering
     \begin{subfigure}[t]{0.5\textwidth}
  \centering
    \includegraphics[clip, trim=0 165 20 0,  scale=1]{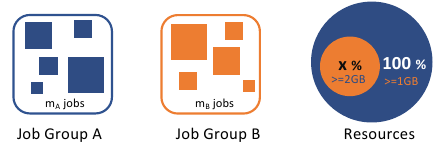} 
       \caption{Resource supply and demand. There are 100\% devices with memory size $\geq$ 1GB and x\% of the devices have memory size $\geq$ 2GB. 
       There are two job groups where Group A with $m_A$ jobs requests for devices with memory  size $\geq$ 1 GB and Group B with $m_B$ jobs requests for devices with memory  size $\geq$ 2 GB.
       }
     \label{fig:step-setup}
  \end{subfigure}
     \begin{subfigure}[t]{0.5\textwidth}
   \centering 
    \includegraphics[clip, trim=0 200 0 0,  scale=1]{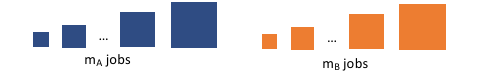} 
    \caption{ Step 1: Sort within Job Group.   }
   \label{fig:amg-step1}
     \end{subfigure}
      \begin{subfigure}[t]{0.5\textwidth}
   \centering 
    \includegraphics[clip, trim=0 166 0 0,  scale=1]{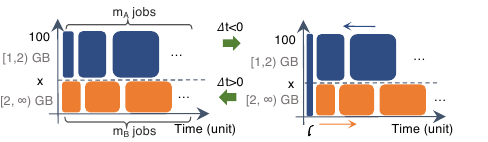} 
    \caption{ Step 2: Schedule across Job Groups.   }
   \label{fig:amg-step2}
     \end{subfigure}
    \caption{\name scheduling algorithm.}
   \label{fig:algo}
\end{figure}

By leveraging the conclusion of Lemma~\ref{lemma:two-groups}, \name can further generalize to the scenario with more than two job groups with arbitrary resource contention patterns.  
Specifically,  \name greedily compares each pair of job groups $(G_j, G_{k})$ following the order.
 For each pair, \name applies the logic proven in Lemma~\ref{lemma:two-groups} to minimize the average scheduling delay between $G_j$ and $G_{k}$.

%\section{Discussions}
%\label{sec:clairvoyant}
%We consider a performance upper bound by comparing with a clairvoyant version of SRSF and \name, both of which are privy to the total number of rounds for each \fl jobs. 
%As shown in Figure~\ref{fig:upper}, the average JCT improvement is notably more substantial—up to a factor of 2.11—when the total number of rounds is known a priori. 
%Importantly, \name demonstrates performance that closely approximates this upper-bound improvement compared with clairvoyant SRSF.
%Furthermore, the clairvoyant version of \name exhibits additional gains in reducing the average JCT.
%
%
%\begin{figure}[]
%   \centering 
%    \begin{subfigure}[t]{0.23\textwidth} 
%      \includegraphics[trim=0 0 0 0,clip,scale=0.3]{google_tiny_clair-jct.pdf} 
%     \caption{Low workload.} 
%   \label{fig:upper-low} 
%     \end{subfigure}
%         \begin{subfigure}[t]{0.23\textwidth} 
%      \includegraphics[trim=0 0 0 0,clip,scale=0.3]{google_huge_clair-jct.pdf} 
%     \caption{High workload.} 
%   \label{fig:upper-high} 
%     \end{subfigure}
%          \caption{Average JCT improvement upper bound. } 
%   \label{fig:upper} 
%\end{figure}